\documentclass[12pt,draftcls,onecolumn]{IEEEtran}
\usepackage{CJK}
\usepackage{amssymb,amsmath,mathrsfs,amsopn,amsfonts,graphicx,color,amsthm}
\usepackage{algorithmic}
\usepackage{graphicx}
\usepackage{mathtools}
\usepackage{subfig}
\usepackage{float}
\usepackage{multirow}
\usepackage{textcomp}
\usepackage{algorithm}
\usepackage{booktabs}
\usepackage{epstopdf}

\newtheorem{remark}{Remark}
\newtheorem{theorem}{Theorem}
\newtheorem{lemma}{Lemma}

\newtheorem{assumption}{Assumption}
\newtheorem{definition}{Definition}
\newtheorem{corollary}{Corollary}

\theoremstyle{condition}
\newtheorem{condition}{Condition}[]

\allowdisplaybreaks

\def\({\Big(}
\def\){\Big)}

\def\a{\alpha}

\def\ba{\begin{array}}
\def\ea{\end{array}}
\def\ban{\begin{eqnarray*}}
\def\ean{\end{eqnarray*}}
\def\bann{\begin{eqnarray*}}
\def\eann{\end{eqnarray*}}
\def\bd{\begin{description}}
\def\ed{\end{description}}
\def\be{\begin{equation}}
\def\ee{\end{equation}}
\def\bna{\begin{eqnarray}}
\def\ena{\end{eqnarray}}

\def\ban{\begin{eqnarray*}}
\def\ean{\end{eqnarray*}}
\def\bna{\begin{eqnarray}}
\def\ena{\end{eqnarray}}
\def\bnaa{\begin{eqnarray}}
\def\enaa{\end{eqnarray}}
\def\bann{\begin{eqnarray*}}
\def\eann{\end{eqnarray*}}

\ifCLASSINFOpdf
\else
\fi
\begin{document}
\begin{CJK}{GBK}{song}

\title{Continuous-Time Decentralized Online Estimation With Additive Noises}
%
%
%

\author{Xiaozheng Fu,~Yan Chen and Tao Li
\thanks{This work was supported by the National Natural Science Foundation
of China under Grant 62261136550. \emph{(Corresponding author: Tao Li.)} }
\thanks{Xiaozheng Fu is with the School of Mathematics and Statistics, Ningbo University, Ningbo 315211, China (e-mail: fuxiaozheng@nbu.edu.cn).}
\thanks{Yan Chen is with the State Key Laboratory of Mathematical Sciences, Academy of
Mathematics and Systems Science, Chinese Academy of Sciences,
Beijing 100190, China (e-mail: yanchen2026@amss.ac.cn).}
\thanks{Tao Li is with the Key Laboratory of Management, Decision and Information
Systems, Institute of Systems Science, Academy of Mathematics and
Systems Science, Chinese Academy of Sciences, Beijing 100190, China (e-mail: litao@amss.ac.cn).}}

%
%

\markboth{Journal of \LaTeX\ Class Files, May~2023}%
{Shell \MakeLowercase{\textit{et al.}}: Bare Demo of IEEEtran.cls for IEEE Journals}
%



\maketitle

\begin{abstract}
We study a decentralized online estimation problem with additive communication noises over the fixed digraph. Each node has a linear measurement of an unknown parameter with random measurement matrices and runs a continuous-time online estimation algorithm.
We transform the convergence analysis of the algorithm into the stability analysis of the non-autonomous linear stochastic differential equation (SDE) with random time-varying coefficients, and develop the asymptotic stability by numerical approximation theory.
Based on the stability results, we show that the algorithm gains can be properly designed to ensure mean square convergence if the measurement matrices and the communication graph satisfy the stochastic spatial-temporal persistence of excitation condition. Furthermore, a special case where the measurement matrices contain a Markov chain is investigated, and the theoretical results are demonstrated by a numerical example.
\end{abstract}

\begin{IEEEkeywords}
Decentralized online estimation, continuous-time algorithm, asymptotic convergence, stochastic differential equation, random time-varying coefficient.
\end{IEEEkeywords}

\section{Introduction}
\label{Asec:introduction}
\IEEEPARstart{T}{he} decentralized parameter estimation of multi-agent systems is widely used in wireless sensor networks \cite{RAV2021TC}, unmanned aerial vehicle formation \cite{DTZ2021TASE} and radar detection \cite{CWA2022AESM}, and therefore, the design and analysis of algorithms have become a hot topic in control theory. Compared with the centralized algorithms with an information fusion center, the decentralized algorithms are more robust in the case of partial sensor failure, and can reduce the communicating and computing costs of sensors.

Up till now, there are some researches on discrete-time decentralized estimation algorithm, e.g. \cite{KMP2013SICON,MWL2019IJRNC,XG2018SICON}.
{Continuous-time signals are widely used in many practical scenarios, such as electrical and speech signals, and their dynamics are often modeled by (stochastic) differential equations according to physical laws \cite{S2001SSBM}.
For example, in radar detection, the target dynamics naturally evolve in continuous time, with key states such as position and velocity varying smoothly over time. 
The continuous-time decentralized online estimation is therefore well aligned with the intrinsic physical characteristics of radar target motion.} 
In recent years, continuous-time {decentralized} estimation algorithms have attracted much attention, e.g.,  \cite{NS2011ASI,CWH2014TAC,ZZ2012TAC,ZGL2022CSL,ZGL2024JSSC}.
Nascimento and Sayed \cite{NS2011ASI} studied the exponential stability of the continuous-time diffusion decentralized least-mean squares algorithm without noise, and assumed that the {measurement} matrices satisfy the persistence of excitation condition, that is, the integral of the {measurement} matrices over a fixed-length time interval has upper and lower bounds.
Chen et al. \cite{CWH2014TAC} studied the uniformly exponential convergence of the continuous-time decentralized cooperative identification algorithm, and required the {measurement} matrix to be uniformly bounded and satisfy the cooperative persistence of excitation condition.
In addition, measurement losses or node sensing failures can be modeled by random measurement  matrices \cite{WLZ2021TIT}.
Zhang and Zhang \cite{ZZ2012TAC} studied the continuous-time decentralized estimation algorithm, where the measurement matrices were assumed to satisfy the global observability with the known expectations, and proved the mean square convergence. 
Zhu et al. \cite{ZGL2022CSL} proposed the least squares algorithm for a single node based on sampling data, where the regression vectors were required to satisfy the Lipschitz condition for all sample paths, and proved the almost surely asymptotic convergence.
Furthermore, they proposed a decentralized least squares algorithm based on sampling data in \cite{ZGL2024JSSC}, where the regression vectors are required to satisfy the Lipschitz condition for all sample paths and the cooperative excitation condition, and proved the almost surely asymptotic convergence.
Note that the {measurement} matrices are deterministic or random with the known expectation or satisfy the Lipschitz condition in the above literature.


Most of the above works supposed that the communication between agents is ideal, that is, agents can receive accurate measurement information from neighbors.
In the realistic network, the communication between nodes is usually interfered by noises.
The additive noise is an important noise in the communication process of sensor networks, which can be used to model the thermal noise \cite{LZ2010TAC}. The characteristic of the additive noise is that its intensity is independent of the states of agents. Until now, some works have been devoted to the decentralized online estimation with additive noises, e.g.,  \cite{ZZ2012TAC,JVB2023SICON,KMR2012TIT}.
For the discrete-time algorithm, Jakovetic et al. \cite{JVB2023SICON} studied the almost sure convergence of the decentralized estimation with a zero-mean additive communication noise, where the measurement matrix is deterministic.
In \cite{KMR2012TIT}, the regressors are assumed to be i.i.d. with the known expectation and the finite second moment, and be independent with the additive communication noise.
For the continuous-time algorithm, Zhang and Zhang \cite{ZZ2012TAC} considered additive communication noises, which are independent with the random measurement matrices with the known expectation.

Most of the above researches considered the effect of the random {measurement}  matrices and the additive communication noises for the continuous-time estimation problem separately. In real networks, different uncertainties may exist at the same time.
For this purpose, we study a continuous-time decentralized cooperative online estimation with the random {measurement} matrices and the additive communication noises. Each node has a linear measurement of an unknown parameter with Markovian switching measurement matrices and runs a continuous-time online estimation algorithm consisting of an innovation term processing the new measurement and a consensus term taking a weighted sum of its estimate and its neighbours' estimates with the additive communication noises.
Based on the algebraic graph theory and the matrix theory, we transform the convergence of the algorithm into the asymptotic stability of linear  stochastic differential equation (SDEs) with random time-varying coefficients.   

The It\^{o} SDEs with random coefficients \cite{Bishop2019,Yong1999,friedman1975stochastic},
in which both random coefficients and Brownian motions are considered,  are closely related to the system control, estimation and filtering, etc., and are widely used in economic, financial, physical and engineering systems \cite{Lim2002,Soong1973}.
Up to now, SDEs with deterministic drift and diffusion coefficients have been extensively studied, e.g.,
\cite{Higham2003,Higham2007,Khasminskii2011,Mao2007,Mao2015}, 
and the ordinary differential equations with random coefficients have been investigated in the past decades, e.g.,
\cite{Blankenship1977,Geman1979,csorgHo2010stability,edsinger1970mean,Zhang2025SCIS}.
However, there is still lack of the study on the asymptotic stability of SDEs with random time-varying coefficients.

The SDEs with the additive noises in the diffusion terms are important and some well-known stochastic processes are the solutions of this kind of equations, such as the Ornstein-Uhlenbeck process, the Brownian bridge process, etc \cite{Mao2007}. The intensity of additive noises is independent of the state of the system.
If the intensity of additive noises does not tend to zero, i.e. the noises are non-decaying, then the solution cannot converge to the trivial solution even if the homogeneous equation is stable. It is possible that the solution is asymptotically stable only if the additive noises are decaying.
The case with non-decaying additive noises  was studied in \cite{Baccouch2016,Buckwar2011,Tanwani2021,JIANG2023890}, and the case with decaying additive noises was studied in \cite{Bishop2019,Hernandez1992,Cruz2020,Zong2018}.
The estimation of the state transition matrix requires that the norm of the difference between the drift coefficient and its limit decay at least at an exponential rate in \cite{Bishop2019}.

{For the study of the estimation problems, 
the cumulative prediction error needs to be minimized to derive the least squares estimate in \cite{ZGL2022CSL,ZGL2024JSSC}.}
For the study of the SDEs, the non-autonomous SDEs and the autonomous SDEs with Markovian switching were discussed in \cite{Mao2007} and \cite{Mao2006}, respectively.
In \cite{Mao2006}, the SDEs with Markovian switching were considered, and the exponential stability and the asymptotic stability in distribution were studied by numerical solutions.
The mean square asymptotic stability of the SDEs with deterministic time-varying coefficients and additive noises was investigated in \cite{Zong2018}. 
{The asymptotically mean square stability of the stochastic Markovian jump systems with the control input were studied in \cite{WSH2023SCIS}.} 
Due to the coexistence of the random coefficients and the additive noises, the above Lyapunov function method and numerical approximate solutions cannot be used to study the asymptotic stability directly in this paper.

{
Considering the coexistence of the random {measurement} matrices and additive communication noises, we devote to dealing with the continuous-time decentralized online estimation by numerical approximation method in this paper. The main contributions of this study are outlined as follows.
\begin{itemize}
  \item In \cite{Zong2018} and \cite{Mao2006}, the non-autonomous SDEs with deterministic time-varying coefficients and the autonomous SDEs with Markovian switching were studied, respectively. Different from \cite{Zong2018,Mao2006}, we consider non-autonomous SDEs with random time-varying coefficients. 
As the analytical solutions of these equations cannot be obtained, we investigate the asymptotic stability of the trivial solution by the numerical approximate solutions.
    We assume that the drift and diffusion coefficients are upper bounded by two squared integrable functions almost surely, and the upper bound of the diffusion coefficient monotonically decreases to zero. 
We give a stochastic persistence of excitation condition that the induced matrix measure of the conditional expectation of the integral concerning the drift coefficient over a fixed-length interval is upper bounded by a sequence whose summation is minus infinity.
    Under the above assumption and condition, we get the asymptotic stability of the true solution and the numerical approximate solutions.
    Especially, we show that if the drift and diffusion coefficients are $\mathcal O\left(1/(t+1)^{\frac{1}{2}+\varepsilon_1}\right)$ and $\mathcal O\left(1/(t+1)^{\frac{1}{2}+\varepsilon_2}\right)$, respectively,
    and the drift coefficient satisfies some persistence of excitation condition, i.e. the induced matrix measure of the conditional expectation of the integral concerning the drift coefficient  over a fixed-length interval is less than a sequence $-c(m)$ satisfying $\liminf\limits_{m\to\infty} c(m)(1+mh\Delta)^{\frac{1}{2}+\varepsilon_1}\!>\!0$,
    then the mean square asymptotic stability of the true solution is achieved.
  \item 
In \cite{NS2011ASI,CWH2014TAC,ZZ2012TAC}, the {measurement} matrices are assumed to be deterministic or their expectations are known. However, these assumptions are often difficult to be satisfied in practical systems, since the observations may be disturbed by uncertainties and their statistical expectations are frequently unavailable in dynamic environments. 
Different from \cite{NS2011ASI,CWH2014TAC,ZZ2012TAC}, we consider random time-varying {measurement} matrices, which aligns more closely with the characteristics of many practical systems.  
  Based on the stability analysis of the SDEs with random time-varying coefficients, we prove that if the measurement matrices and the graph satisfy the stochastic spatial-temporal persistence of excitation condition, then the algorithm gains can be designed to guarantee mean square convergence of the continuous-time algorithm.
     We further obtain the mean square convergence of the continuous-time algorithm for the case where the {measurement} matrices contain a Markov chain with strongly 1-exponential ergodicity.
\end{itemize}}

This paper is organized as follows:
the continuous-time decentralized online estimation algorithms and the construction of the
the SDEs with random time-varying coefficients 
are presented in Section \ref{sec:probform}.
The asymptotic stability of the SDEs is given in Section \ref{sec:stability}.
The asymptotic convergence of the algorithm is given in Section \ref{sec:convergence}.
A numerical example is given in Section \ref{sec:example}.
The whole paper is concluded in Section \ref{sec:conclusions}.

Symbols and notations: 
$\mathbb{R}^+$: set of positive real numbers; 
$\mathbb R^n$: $n$-dimensional real vector space; $\mathbb{R}^{m\times n}$: $m\times n$-dimensional real matrix space; 
$\lfloor x\rfloor$: biggest integer less than or equal to $x$; 
$|a|$: absolute value of the real number $a$; 
$\mathbf{0}_n$: $n$-dimensional zero vector; 
$\mathbf{1}_N$: $N$-dimensional vector with all elements being $1$; 
$\mathbf{0}_{n\times n}$: $n$-dimensional zero matrix;
$I_n$: $n$-dimensional identity matrix;
$\otimes$: Kronecker product; 
$\textbf{diag}(A_1,\cdots,A_n)$: block diagonal matrix whose diagonal elements are $A_1,\cdots,A_n$; 
$A^{\top}$: transpose of matrix $A$; 
$\lambda_{\max}(A)$: maximum eigenvalue of the real symmetric matrix $A$; 
$\lambda_{\min}(A)$: minimum eigenvalue of the real symmetric matrix $A$; 
$\rho(A)$: spectral radius of matrix $A$; 
$\|A\|_2$: 2-norm of matrix $A$; $\mu_2(A)=\lambda_{\max}\left(\frac{A+A^{\top}}{2}\right)$: matrix measure induced by the 2-norm; 
$(\Omega, \mathcal F, P)$: complete probability space; $\{\mathcal F(t): t\geq 0\}$: $\sigma$-algebraic flow on $(\Omega, \mathcal F, P)$ satisfying the usual conditions that $\mathcal{F}(t)$ is right continuous and $\mathcal{F}(0)$ contains all zero probability sets; 
$\mathbb{E}[\xi]$: mathematical expectation of the random variable $\xi$; 
$L_{\mathcal{F}(t)}^p(\Omega, \mathbb{R}^{n\times n})$: the family of $\mathbb{R}^{n\times n}$-valued $\mathcal{F}(t)$-measurable random variables $\xi$ with $\mathbb{E}[\|\xi\|_2^p]<\infty$; 
$a_n=\mathcal O(b_n)$: $\lim\sup_{n\to\infty}\frac{|a_n|}{b_n}<\infty$, where $\{a_n,n\ge 0\}$ is the real sequence and $\{b_n,n\ge 0\}$ is the sequence of positive real numbers; 
$a_n=o(b_n)$: $\lim_{n\to\infty}\frac{a_n}{b_n}=0$; 
for a sequence of $n$-dimensional matrices $\{Y(k),k\geq 0\}$, denote
$$
\Phi_Y(j,i)=\left\{
\begin{aligned}
&Y(j)\cdots Y(i), \quad & j\geq i,  \\
&I_n, \quad & j<i.
\end{aligned}
\right.
$$
For a series of scalars $\{c(k), k\geq 0\}$, denote
$$ \prod_{k=i}^jc(k)=\left\{
\begin{aligned}
&c(j)\cdots c(i), \ & j\geq i,  \\
&1, \ & j<i,
\end{aligned}
\right.
\quad
\sum_{k=i}^jc(k)=\left\{
\begin{aligned}
&c(i)+\cdots+ c(j), \ & j\geq i,  \\
&0, \ & j<i.
\end{aligned}
\right.
$$

\section{Problem Formulations}\label{sec:probform}

Consider a balanced fixed digraph consisting of $N$ nodes. Suppose that the relationships among nodes are described by the graph $\mathcal{G}=\{\mathcal{V}, \mathcal{E}_{\mathcal{G}}, \mathcal{A}_{\mathcal{G}}\}$,
where $\mathcal{V}=\{1,2,\cdots, N\}$ is the set of nodes, $\mathcal{E}_{\mathcal{G}}$ is the set of edges, and $\mathcal{A}_{\mathcal{G}}=[a_{ij}]_{i,j=1}^N$ is the adjacency matrix.
Denote the neighbors of the $i$th node by $\mathcal{N}_i=\{j\in \mathcal{V}|(j,i)\in\mathcal{E}_{\mathcal{G}}\}$, the degree matrix of $\mathcal{G}$ by  $\mathcal{D}_{\mathcal{G}}=\textbf{diag}\left(\sum_{j=1}^Na_{1j},\sum_{j=1}^Na_{2j},\cdots,\sum_{j=1}^N a_{Nj}\right)$, and the Laplacian matrix by  $\mathcal{L}_{\mathcal{G}}=\mathcal{D}_{\mathcal{G}}-\mathcal{A}_{\mathcal{G}}$.
All nodes over the network cooperatively estimate the
unknown parameter vector {$\theta$}  by information exchange among nodes. For each node $i\in\{1, 2, \cdots, N\}$, we assume that its measurement of {$\theta$} is a linear function, i.e. the measurement of node $i$ at instant $t$ satisfies
{\setlength\abovedisplayskip{3pt}
\setlength\belowdisplayskip{3pt}
\bna\label{model1}
&&\hspace{-0.4cm}dz_i(t)=H_i(t){\theta}dt,
~i=1,2,\cdots,N,~t\in\mathbb{R}^+,
\ena
where $H_i(t)\in\mathbb{R}^{n_i\times n}(n_i\leq n)$ are the random measurement matrices.}

For node $i$, we consider the following continuous-time decentralized cooperative online estimation algorithm
{\setlength\abovedisplayskip{3pt}
\setlength\belowdisplayskip{1pt}
\bna\label{model3}
&&\hspace{-1.4cm}d{\theta_i(t)}=\alpha(t)H_i^{\top}(t)[dz_i(t)-H_i(t){\theta_i(t)}dt]
+\beta(t)\sum\limits_{j\in {\mathcal{N}_i}}a_{ij}{dy_{ji}(s)},
\ena
{with
$
dy_{ji}(t)=({\theta_j(t)-\theta_i(t)})dt+\sigma_{ji}dw_{ji}(t) 
$ 
denoting the measurement of relative states by agent $i$ from its
neighbor $j\in\mathcal{N}_i$.} Here, ${\theta_{i}(t)}\in\mathbb{R}^n$ is the state of node $i$, representing its estimate of {$\theta$},    $\{w_{ji}(t),i,j=1,2,\cdots, N\}$ are independent Brownian motions,  $\sigma_{ji}\in\mathbb{R}^n$ is the intensity coefficient of the additive measurement noise, and $\alpha(t)$ and $\beta(t)$ are the algorithm gains.}
{The detailed steps of the continuous-time decentralized online estimation algorithm are summarized in Algorithm \ref{alg:CTDOEA}.}


Denote the $\sigma$-field
$\mathcal{F}(t)=\sigma(H_i(s),w_{ji}(s),j,i=1,2,\cdots, N,0\leq s\leq t)$,
\begin{align*}
&z(t)=[z_1^{\top}(t), z_2^{\top}(t), \cdots, z_N^{\top}(t)]^{\top},
H(t)=[H_1^{\top}(t), H_2^{\top}(t), \cdots, H_N^{\top}(t)]^{\top},\\
&{\Theta(t)=[\theta_1^{\top}(t),\theta_2^{\top}(t),\cdots,\theta_N^{\top}(t)]^{\top}},
D=\textbf{diag}\{\alpha_1^{\top}\otimes I_n, \alpha_2^{\top}\otimes I_n,\cdots,\alpha_N^{\top}\otimes I_n\},\\
&\mathcal{H}(t)=\textbf{diag}\{H_1(t),H_2(t),\cdots,H_N(t)\},
\Sigma=\textbf{diag}\{\sigma_{11},\cdots,\sigma_{N1}, \cdots,  \sigma_{1N}, \cdots, \sigma_{NN}\},\\
&w(t)=[w_{11}(t), \cdots, w_{N1}(t),\cdots, w_{1N}(t),\cdots, w_{NN}(t)]^{\top},
\end{align*}
where $\alpha_i^{\top}$ is the $i$-th row of $\mathcal{A}_{\mathcal{G}}$.
Then (\ref{model1}) can be written as the following compact form
{\setlength\abovedisplayskip{5pt}
\setlength\belowdisplayskip{5pt}
\bna\label{model2}
&&\hspace{-0.4cm}dz(t)=H(t){\theta}dt,
~t\in\mathbb{R}^+,
\ena
and (\ref{model3}) can be written as the following compact form}
{\setlength\abovedisplayskip{5pt}
\setlength\belowdisplayskip{5pt}
\ban
&&\hspace{-0.5cm}d{\Theta(t)}=\big[-\alpha(t)\mathcal{H}^{\top}(t)\mathcal{H}(t)-\beta(t)(\mathcal{L}_{\mathcal{G}}\otimes I_n)\big]{\Theta(t)}dt+\alpha(t)\mathcal{H}^{\top}(t)dz(t)
+\beta(t)D\Sigma dw(t).
\ean
Denote the estimation error by $e(t)={\Theta(t)}-\mathbf{1}_N\otimes {\theta}$. From the above equation, (\ref{model2}), $(\mathcal{L}_{\mathcal{G}}\otimes I_n)(\mathbf{1}_N\otimes {\theta})=0$ and $\mathcal{H}(t)(\mathbf{1}_N\otimes {\theta})=H(t){\theta}$, we have}
{\setlength\abovedisplayskip{5pt}
\setlength\belowdisplayskip{3pt}
\begin{align}
de(t)
= &[-\alpha(t)\mathcal{H}^{\top}(t)\mathcal{H}(t)-\beta(t)(\mathcal{L}_{\mathcal{G}}\otimes I_n)]e(t)dt+[-\alpha(t)\mathcal{H}^{\top}(t)\mathcal{H}(t)\notag\\
&-\beta(t)(\mathcal{L}_{\mathcal{G}}\otimes I_n)]
(\mathbf{1}_N\otimes {\theta})dt+\alpha(t)\mathcal{H}^{\top}(t)H(t){\theta}dt+\beta(t)D\Sigma dw(t)\notag\\
=& [-\alpha(t)\mathcal{H}^{\top}(t)\mathcal{H}(t)-\beta(t)(\mathcal{L}_{\mathcal{G}}\otimes I_n)]e(t)dt+\beta(t)D\Sigma dw(t).\label{errorequation}
\end{align}}

The above error equation (\ref{errorequation}) comes down to the following linear SDEs with random time-varying coefficients and the additive noise:
%
\bna\label{eq1}
\left\{
\begin{array}{l}
dx(t)=A(t,\omega)x(t)dt+D(t,\omega)dw(t), ~~t\geq 0,\\
x(0)=x_0,
\end{array}\right.
\ena
where $x(t)\in\mathbb{R}^n$ is the state, $(w(t)\in\mathbb{R}^m, \mathcal{F}(t))$ is an $m$-dimensional standard Wiener process, $(A(t,\omega)\in\mathbb{R}^{n\times n}, \mathcal{F}(t))$ and $(D(t,\omega)\in\mathbb{R}^{n\times m}, \mathcal{F}(t))$ are matrix-valued adapted processes.
For convenience, the sample point $\omega$ is omitted.
\vskip 0.2cm

\begin{remark}
\rm{There are many results on the moment stability of SDEs, e.g.,  \cite{Higham2007,Mao2007,Mao2006,Zong2018}. In \cite{Zong2018}, the mean square asymptotic stability of the non-autonomous SDEs with deterministic time-varying coefficients {was} studied. In \cite{Mao2006}, the autonomous SDEs with Markovian switching were considered, and the exponential stability and the asymptotic stability in distribution were studied by numerical solutions. Motivated by the above results, we develop the mean square asymptotic stability of non-autonomous SDEs with random coefficients that can be used to solve decentralized estimation problems.}
\end{remark}

As the analytical solutions of (\ref{eq1}) cannot be obtained, we propose a numerical approximation method to construct new numerical approximate solutions and investigate the asymptotic stability of the trivial solution by the numerical approximate solutions.
\label{construction}
The equivalent integral form of (\ref{eq1}) is given by
{\setlength\abovedisplayskip{5pt}
\setlength\belowdisplayskip{5pt}
\ban
&&\hspace{-0.4cm}x(t)=x_0+\int_{0}^tA(s)x(s)ds+\int_{0}^tD(s)dw(s),~t\geq 0.
\ean
Given a step $\Delta>0$, define the following DTNAS}
{\setlength\abovedisplayskip{5pt}
\setlength\belowdisplayskip{5pt}\ban
&&\hspace{-0.4cm}X((k+1)\Delta)=X(k\Delta)+\int_{k\Delta}^{(k+1)\Delta}A(s)ds X(k\Delta)+\int_{k\Delta}^{(k+1)\Delta}D(s)dw(s),
\ean
i.e.}
\bna\label{EM5}
&&\hspace{-0.4cm}X((k+1)\Delta)=(I_n+\widetilde{A}(k+1))X(k\Delta)+\xi(k+1),~k=0,1,\cdots,
\ena
where $\widetilde{A}(k)=\int_{(k-1)\Delta}^{k\Delta}A(s)ds, \xi(k)=\int_{(k-1)\Delta}^{k\Delta}D(s)dw(s)$, $k=1,2,\cdots$, $\widetilde{A}(0)=\mathbf{0}_{n\times n}$, $\xi(0)=\mathbf{0}_{n}$.
Then $\{\widetilde{A}(k),\mathcal{F}(k\Delta), k=1,2,\cdots\}$ and $\{\xi(k),\mathcal{F}(k\Delta), k=1,2,\cdots\}$ are adapted sequences.
Denote $\overline{X}(t)=\sum_{k=0}^\infty X(k\Delta)I_{[k\Delta, (k+1)\Delta)}(t), t\geq 0$, and $n_{\Delta,T}=\lfloor\frac{T}{\Delta}\rfloor$, $T\geq 0$. Define the following CTNAS
{\setlength\abovedisplayskip{5pt}
\setlength\belowdisplayskip{5pt}
\bna\label{EM6}
&&\hspace{-0.4cm}\widehat{X}_{n_{\Delta,T}\Delta}(t)=x(n_{\Delta,T}\Delta)+\int_{n_{\Delta,T}\Delta}^tA(s)\overline{X}(s)ds+\int_{n_{\Delta,T}\Delta}^tD(s)dw(s),
\ena
with 
$\widehat{X}_{n_{\Delta,T}\Delta}(n_{\Delta,T}\Delta)=x(n_{\Delta,T}\Delta)$.
Obviously, $X(k\Delta)=\widehat{X}_{n_{\Delta,T}\Delta}(k\Delta)=\overline{X}(k\Delta)$, $k\geq n_{\Delta,T}$.
Specially, if $T=0$, $\widehat{X}_{n_{\Delta,T}\Delta}(t)$ degenerates to the CTNAS in \cite{Mao2007}.}
\vskip 0.2cm
In the following, we give the definitions of the mean square asymptotic stabilities of the solution of (\ref{eq1}), the DTNAS (\ref{EM5}) and the CTNAS (\ref{EM6}), respectively.

\begin{definition}[{See \cite{Mao2007}}]\label{definition1}
\rm{The solution of (\ref{eq1}) is asymptotically stable in mean square, if $\lim_{t\to \infty}\\\mathbb{E}\big[\|x(t)\|_2^2\big]=0$ for any initial value ${x_0}\in L_{\mathcal{F}(0)}^2(\Omega, \mathbb{R}^n)$.}
\end{definition}
                                                                
\begin{definition}\label{definition2}
\rm{For a given step $\Delta>0$, the discrete-time $\Delta-$numerical approximate solution (\ref{EM5}) is mean square asymptotically stable, if $\lim_{k\to \infty}\mathbb{E}\Big[\|X(k\Delta)\|_2^2\Big]=0$ for any initial value ${x_0}\in L_{\mathcal{F}(0)}^2(\Omega, \mathbb{R}^n)$.}
\end{definition}

\begin{definition}\label{definition3}
\rm{For a given step $\Delta>0$ and $T\geq 0$, the continuous-time $\Delta-$numerical approximate solution (\ref{EM6}) is mean square asymptotically stable, if $\lim_{t\to \infty}\mathbb{E}\Big[\big\|\widehat{X}_{n_{\Delta,T}\Delta}(t)\big\|_2^2\Big]=0$ for any initial value $x(n_{\Delta,T}\Delta)\in L_{\mathcal{F}(n_{\Delta,T}\Delta)}^2(\Omega, \mathbb{R}^n)$.}
\end{definition}

%
%
%
%
%
%
%
%

\begin{algorithm}[!h] 
    \caption{{Continuous-time decentralized online estimation algorithm}}
    \label{alg:CTDOEA}
    \renewcommand{\algorithmicrequire}{\textbf{Input:}}
    \renewcommand{\algorithmicensure}{\textbf{Output:}}
    
    \begin{algorithmic}[1]
        \REQUIRE  measurement matrices $H_i(t)$; measurements $z_i(t)$; 
gains $\alpha(t),\beta(t)$; adjacency matrix $\mathcal{A}_{\mathcal G}$;
measurements of relative states $y_{ji}(t)$;
step size $\Delta>0$.

        \STATE Initialize estimates ${\theta_i(0)}$

\FOR{$k = 0,1,2,\ldots$}
    \STATE $t_k \gets k\Delta$
    \FOR{$i = 1,2,\ldots,N$}


        \STATE \textbf{Compute local innovation term:}
        \STATE $g_i(t_k) \gets \int_{t_k}^{t_{k+1}}\alpha(s)H_i^{\top}(s)dz_i(s) - \int_{t_k}^{t_{k+1}}\alpha(s)H_i^{\top}(s)H_i(s)x_i(s)ds$

        \STATE \textbf{Compute consensus coupling term:}
        \STATE $c_i(t_k) \gets \int_{t_k}^{t_{k+1}}\beta(s)\sum_{j \in \mathcal{N}_i} a_{ij}dy_{ji}(s)$

        \STATE \textbf{Update the estimates:}
        \STATE ${\theta_i(t_{k+1})} \gets {\theta_i(t_k)}
        + g_i(t_k) + c_i(t_k)$

    \ENDFOR
\ENDFOR
        \ENSURE Estimates ${\theta_i(t)}$.
    \end{algorithmic}
\end{algorithm}
\section{Asymptotic stabilities of the solutions}
\label{sec:stability}

%
%

We have the following assumption on the coefficients in (\ref{eq1}).
\vskip 0.2cm
\begin{assumption}\label{as1}
\rm{$\{A(t),t\geq 0\}$ is independent of $\{D(t), w(t), t\geq 0\}$ and
there exist real-valued functions $a(t)$ and $d(t)$
with $\int_{0}^\infty a^2(t)dt<\infty$, $\int_{0}^\infty d^2(t)dt<\infty$, $d(t)$ monotonically decreases, and $d(t+\Delta)=\mathcal{O}(d(t))$, for a given $\Delta>0$, such that $\|A(t)\|_2\leq a(t), \|D(t)\|_2\leq d(t), \forall~t\geq 0~\rm{a.s.}$}
\end{assumption}
\vskip 0.1cm

\begin{remark}
\rm{
(i) The independence assumption in Assumption \ref{as1}, which requires
$\{A(t), t \ge 0\}$ to be independent of $\{D(t), w(t), t \ge 0\}$ in (\ref{eq1}), is reasonable.
Since (\ref{eq1}) includes (\ref{errorequation}) as a special case, for (\ref{errorequation}), this assumption is equivalent to requiring that
$\{\mathcal{H}(t), t \ge 0\}$ is independent of $\{w(t), t \ge 0\}$. 
This means that the randomness in the observation matrices is independent of the communication noises. 
This is natural because these two types of uncertainties originate from different sources:
the randomness in the observation matrices is generated locally by the sensing process, whereas the communication noises
are introduced by the network during information transmission.
Hence, there is no direct coupling between them, and assuming independence is well justified.
Note that similar assumptions are standard in the discrete-time decentralized parameter estimation literature \cite{ZZ2012TAC,ZTS2012TSP}. 
(ii) In Assumption \ref{as1}, the assumption that both the drift coefficient and the diffusion coefficient are bounded by a decaying sequence is also reasonable.
For (\ref{model3}), our objective is to ensure that the estimation error converges to zero, and thus it is natural to require the right-hand side of (\ref{model3}) to decay to zero as $t \to \infty$. 
Moreover, this assumption can be satisfied in practice by properly designing time-varying gains $\alpha(t)$ and $\beta(t)$.
}
\end{remark}

Then, we give the following conditions based on the drift coefficient, where Condition \ref{condition1} (i) is called a stochastic persistence of excitation condition.
\vskip -0.2cm
\begin{condition}\label{condition1}
For a given $\Delta>0$, there exists an integer $h>0$, a positive real sequence $\{c(m),m\geq 0\}$ tending to zero with $\sum_{m=0}^\infty c(m)=\infty$, 
and a positive real sequence $\{\rho(m),m\geq 0\}$ monotonically decreasing to zero with
$d(m\Delta)=\mathcal O(\rho(m))$ and $\rho^2(mh)=o(c(m))$, such that \\
\rm{(i)}~ $\mu_2\left(\mathbb{E}\Big[\int_{mh\Delta}^{(m+1)h\Delta}A(s)ds\Big|\mathcal{F}(mh\Delta)\Big]\right) \leq -c(m)$~\rm{a.s.},$~m=0,1,\cdots$,\\
\rm{(ii)}~ 
$\mathbb{E}\bigg[\bigg(\max\limits_{k\Delta\leq s< (k+1)\Delta}\|A(s)\|_2\bigg)^{2^{\max\{h,2\}}}\bigg|\mathcal{F}(k\Delta)\bigg]^{\frac{1}{2^{\max\{h,2\}}}}\leq \rho(k)$~\rm{a.s.},$~k=0,1, \cdots$.\\
\end{condition}



To prove the asymptotic stability of (\ref{eq1}), we need the following four steps.
Among the results, the proofs of Lemmas \ref{lemma5}, \ref{lemma2}, \ref{lemma6} and \ref{lemma3} are given in Appendix \ref{appendix:A}.
\vskip 0.2cm
(I) We give the conditions for the mean square asymptotic stability of the DTNAS (\ref{EM5}), and the conditions for the equivalence of the mean square asymptotic stabilities between the DTNAS (\ref{EM5}) and the CTNAS (\ref{EM6}).
%
%
\vskip 0.2cm

Firstly, we give sufficient conditions for the mean square asymptotic stability of the DTNAS (\ref{EM5}).
Denote $V(k\Delta)=X^{\top}(k\Delta)X(k\Delta)$ and $m_k=\lfloor\frac{k}{\Delta}\rfloor$.

\begin{lemma}\label{lemma5}
\rm{For (\ref{EM5}), if there exists a constant $\Delta>0$ such that Assumption \ref{as1} and Condition \ref{condition1} hold,
then the DTNAS (\ref{EM5}) is asymptotically stable in mean square and
{\setlength\abovedisplayskip{7pt}
\setlength\belowdisplayskip{7pt}
\begin{align*}
&~\mathbb{E}[V((k+1)\Delta)]\cr
\leq&~ \Big(1+2\Delta\rho(m_kh)+\Delta^2\rho^2(m_kh)\Big)^h\prod_{i=L}^{m_k-1}\bigg(1-\frac{1}{2}c(i)\bigg)\mathbb{E}[V(Lh\Delta)]\cr
&~+2h\Delta\Big(1+2\Delta\rho(m_kh)+\Delta^2\rho^2(m_kh)\Big)^h\sum_{i=L}^{m_k-1}d^2(ih\Delta) \prod_{j=i+1}^{m_k-1}\bigg(1-\frac{1}{2}c(j)\bigg)\cr
&~+h\Delta d^2(m_kh\Delta)\Big(1+2\Delta\rho(m_kh)+\Delta^2\rho^2(m_kh)\Big)^h,~L\geq0.
\end{align*}}}
\end{lemma}

In the following, we give sufficient conditions for the equivalence of the mean square asymptotic stabilities between the DTNAS (\ref{EM5}) and the CTNAS (\ref{EM6}).

\vskip 0.2cm
\begin{lemma}\label{lemma2}
\rm{For (\ref{eq1}), if there exists $\Delta>0$ such that Assumption \ref{as1} holds,
then the mean square asymptotic stabilities of the DTNAS (\ref{EM5}) and the CTNAS (\ref{EM6})  are equivalent.}
\end{lemma}
\vskip 0.2cm

(II) Based on Lemma \ref{lemma5} and Lemma \ref{lemma2}, we give sufficient conditions for the mean square asymptotic stability of the CTNAS (\ref{EM6}) and the mean square convergence rate.
\begin{lemma}\label{lemma6}
\rm{For (\ref{EM6}),
if there exists $\Delta>0$ such that Assumption \ref{as1}  and Condition \ref{condition1} hold,
then the CTNAS (\ref{EM6}) is mean square asymptotically stable and
{\setlength\abovedisplayskip{9pt}
\setlength\belowdisplayskip{9pt}
\begin{align*}
&~\mathbb{E}\left[\Big\|\widehat{X}_{n_{\Delta,T}\Delta}(t)\Big\|_2^2\right]\cr
\leq&~
\varpi(t)\Bigg[
\mathbb{E}[V(n_{\Delta,T}\Delta h)]
+2\int_{\frac{n_{\Delta,T}}{h}}^{\frac{t}{h\Delta}}\iota(sh\Delta)\exp\Bigg(-\int_{s+1}^{\frac{t}{h\Delta}-4}\vartheta(\tau)d\tau\Bigg)ds\cr
&~+C\iota(t-(h+1)\Delta)\Bigg],~t\geq n_{
\Delta,T}\Delta,~T\geq4 h\Delta,
\end{align*}
where}
{\setlength\abovedisplayskip{8pt}
\setlength\belowdisplayskip{3pt}
\begin{align}\label{l6-1}
&\vartheta(t)=\frac{1}{2}c(k),~t\in[k\Delta,(k+1)\Delta),~k\geq n_{\Delta,T},\cr
&\iota(t)=h\Delta d^2(k\Delta),~t\in[k\Delta,(k+1)\Delta),~k\geq n_{\Delta,T},\cr
&\varpi(t)=3\Bigg(1+n\Delta\int_{k\Delta}^{(k+1)\Delta}a^2(s)ds\Bigg)
\big(1+2\Delta\rho(k-h-1)
+\Delta^2\rho^2(k-h-1)\big)^h,\cr
&\qquad \qquad \qquad \qquad \qquad \qquad \qquad \qquad \qquad \qquad ~~t\in[k\Delta,(k+1)\Delta),~k\geq n_{\Delta,T}.
\end{align}}}
\end{lemma}

(III) In the following, we give an estimate of the mean square upper bound of the difference between the true solution of (\ref{eq1}) and the CTNAS (\ref{EM6}). 
\vskip 0.2cm
\begin{lemma}\label{lemma3}
\rm{For any $T\geq 0$, if there exists a constant $\Delta>0$ such that Assumption \ref{as1} holds, then the solution of (\ref{eq1}) and the CTNAS $\widehat{X}_{n_{\Delta,T}\Delta}(t)$ satisfy
{\setlength\abovedisplayskip{5pt}
\setlength\belowdisplayskip{5pt}
\bna\label{L3-1}
&&\hspace{0.2cm}\sup\limits_{n_{\Delta,T}\Delta\leq t\leq T'}\mathbb{E}\Bigg[\Big\|\widehat{X}_{n_{\Delta,T}\Delta}(t)-x(t)\Big\|_2^2\Bigg]\cr
&&\hspace{-0.4cm}\leq \Bigg[4n(T'-n_{\Delta,T}\Delta)\int_{n_{\Delta,T}\Delta}^{T'}a^2(s)\Bigg(n\Delta\int_{n_{\Delta,s}\Delta}^{(n_{\Delta,s}+1)\Delta}a^2(\tau)d\tau
\sup\limits_{n_{\Delta,T}\Delta\leq r\leq T'}\mathbb{E}\Bigg[\Big\|\widehat{X}_{n_{\Delta,T}\Delta}(r)\Big\|_2^2\Bigg]\cr
&&\hspace{-0.0cm}+\int_{n_{\Delta,s}\Delta}^{(n_{\Delta,s}+1)\Delta}d^2(\tau)d\tau\Bigg)ds\Bigg] 
\exp\Bigg(2n(T'-n_{\Delta,T}\Delta)\int_{n_{\Delta,T}\Delta}^{T'}a^2(s)ds\Bigg),\cr
&&\hspace{8.2cm}
~\forall~T'\geq n_{\Delta,T}\Delta\geq 0.
\ena}}
\end{lemma}

(IV)
Finally, based on the mean square convergence rate of the CTNAS obtained in Lemma \ref{lemma6} and the estimate of the mean square upper bound of the difference between the true solution of (\ref{eq1}) and the CTNAS (\ref{EM6}) in Lemma \ref{lemma3},
we give sufficient conditions for the mean square asymptotic stability of the solution of (\ref{eq1}). 
For the simplicity of the description, denote
{\setlength\abovedisplayskip{5pt}
\setlength\belowdisplayskip{2pt}
\begin{align*}
D_{k}=&~
(T_{k+1}-T_{k})\int_{T_{k}\Delta}^{T_{k+1}\Delta}a^2(s)
\int_{n_{\Delta,s}\Delta}^{(n_{\Delta,s}+1)\Delta}\big(a^2(\tau)+d^2(\tau)\big)d\tau ds\exp\bigg(2n\Delta(T_{k+1}
-T_{k})\cr
&\times
\int_{T_{k}\Delta}^{T_{k+1}\Delta}a^2(s)ds\bigg)
+\int_{\frac{T_{k}}{h}}^{\frac{T_{k+1}}{h}}\exp\bigg(-\int_{s+1}^{\frac{T_{k}}{h}-4}\vartheta(\tau)d\tau\bigg)
\iota(sh\Delta)ds,\cr
&~\text{where}~ \vartheta(t)=\frac{1}{2}c(k), \iota(t)=h\Delta d^2(k\Delta), t\in[k\Delta,(k+1)\Delta),~k\geq \left\lfloor\frac{T}{\Delta}\right\rfloor. 
\end{align*}}
\begin{theorem}\label{theorem1}
\rm{For (\ref{eq1}), if there exists a constant $\Delta>0$ such that Assumption \ref{as1}  and Condition \ref{condition1} hold, 
and there exists an integer sequence $\{T_k\}_{k\geq 1}$ monotonically increasing to infinity with $T_1\geq 4h$, 
such that
$\lim_{k\to\infty}D_{k}=0,$
then the solution of (\ref{eq1}) is mean square asymptotically stable.}
\end{theorem}

\proof
Let $\widehat{X}_{T_k\Delta}(t)$ be the CTNAS generated by the initial value $x(T_{k}\Delta)$. Then by Lemma \ref{lemma6}, we get
{\setlength\abovedisplayskip{7pt}
\setlength\belowdisplayskip{7pt}
\bna\label{T1-3}
&&\hspace{-0.0cm}\sup\limits_{T_{k}\Delta\leq t\leq T_{k+1}\Delta}\mathbb{E}\left[\Big\|\widehat{X}_{T_{k}\Delta}(t)\Big\|_2^2\right]\cr
&&\hspace{-0.4cm}\leq
\varpi(T_{k}\Delta)\Bigg[
\mathbb{E}[V(T_{k}\Delta h)]+2\int_{\frac{T_k}{h}}^{\frac{T_{k+1}}{h}}\iota(sh\Delta)
\exp\Bigg(-\int_{s+1}^{\frac{T_{k}}{h}-4}\vartheta(\tau)d\tau\Bigg)ds\cr
&&\hspace{-0.0cm}
+C\iota(T_{k}\Delta-(h+1)\Delta)\Bigg].
\ena
From the basic inequality $(x+y)^2\leq(1+\alpha)x^2+\Big(1+\frac{1}{\alpha}\Big)y^2$, for any $\alpha>0, x, y\in \mathbb{R}$, we have}
{\setlength\abovedisplayskip{5pt}
\setlength\belowdisplayskip{5pt}
\ban
&&\hspace{-0.6cm}\mathbb{E}\Big[\|x(t)\|_2^2\Big]\leq (1+\alpha)\mathbb{E}\left[\Big\|x(t)-\widehat{X}_{T_{k}\Delta}(t)\Big\|_2^2\right]+\Bigg(1+\frac{1}{\alpha}\Bigg)\mathbb{E}\left[\Big\|\widehat{X}_{T_{k}\Delta}(t)\Big\|_2^2\right].
\ean
This together with Lemma \ref{lemma3} and (\ref{T1-3}) yields}
{\setlength\abovedisplayskip{6pt}
\setlength\belowdisplayskip{6pt}
\begin{align}\label{T1-4}
&~\sup\limits_{T_{k}\Delta\leq t\leq T_{k+1}\Delta}\mathbb{E}\Big[\|x(t)\|_2^2\Big]\cr
\leq&~ (1+\alpha)\sup\limits_{T_{k}\Delta\leq t\leq T_{k+1}\Delta}\mathbb{E}\left[\Big\|x(t)-\widehat{X}_{T_{k}\Delta}(t)\Big\|_2^2\right]
+\Bigg(1+\frac{1}{\alpha}\Bigg)\sup\limits_{T_{k}\Delta\leq t\leq T_{k+1}\Delta}\mathbb{E}\left[\Big\|\widehat{X}_{T_{k}\Delta}(t)\Big\|_2^2\right]\cr
\leq&~ (1+\alpha)\Bigg[4n(T_{k+1}\Delta-T_{k}\Delta)\int_{T_k\Delta}^{T_{k+1}\Delta}a^2(s)\Bigg(n\Delta\int_{n_{\Delta,s}\Delta}^{(n_{\Delta,s}+1)\Delta}a^2(\tau)d\tau\cr
&~\times\sup\limits_{T_k\Delta\leq r\leq T_{k+1}\Delta}\mathbb{E}\left[\Big\|\widehat{X}_{T_{k}\Delta}(r)\Big\|_2^2\right]+\int_{n_{\Delta,s}\Delta}^{(n_{\Delta,s}+1)\Delta}d^2(\tau)d\tau\Bigg)ds\Bigg]
\exp\Bigg(2n(T_{k+1}\Delta\cr
&~-T_{k}\Delta)
\int_{T_{k}\Delta}^{T_{k+1}\Delta}a^2(s)ds\Bigg)
+\Bigg(1+\frac{1}{\alpha}\Bigg)\sup\limits_{T_{k}\Delta\leq t\leq T_{k+1}\Delta}\mathbb{E}\left[\Big\|\widehat{X}_{T_{k}\Delta}(t)\Big\|_2^2\right]\cr
\leq&~ (1+\alpha)4n(T_{k+1}\Delta-T_{k}\Delta)\int_{T_{k}\Delta}^{T_{k+1}\Delta}a^2(s)\Bigg[n\Delta\int_{n_{\Delta,s}\Delta}^{(n_{\Delta,s}+1)\Delta}a^2(\tau)d\tau
\varpi(T_{k}\Delta)\cr
&~\times\Bigg[\mathbb{E}[V(T_k\Delta h)]
+2\int_{\frac{T_{k}}{h}}^{\frac{T_{k+1}}{h}}\iota(rh\Delta)\exp\Bigg(-\int_{r+1}^{\frac{T_{k}}{h}-4}\vartheta(\tau)d\tau\Bigg)dr\cr
&~+C\iota(T_k\Delta-(h+1)\Delta)\Bigg]
+\int_{n_{\Delta,s}\Delta}^{(n_{\Delta,s}+1)\Delta}d^2(\tau)d\tau\Bigg]ds \exp\Bigg(2n(T_{k+1}\Delta-T_{k}\Delta)\cr
&~\times\int_{T_{k}\Delta}^{T_{k+1}\Delta}a^2(s)ds\Bigg)
+\Bigg(1+\frac{1}{\alpha}\Bigg)\varpi(T_{k}\Delta)
\Bigg[\mathbb{E}[V(T_k\Delta h)]\cr
&~+2\int_{\frac{T_{k}}{h}}^{\frac{T_{k+1}}{h}}\iota(sh\Delta)\exp\Bigg(-\int_{s+1}^{\frac{T_{k}}{h}-4}\vartheta(\tau)d\tau\Bigg)ds+C\iota(T_{k}\Delta-(h+1)\Delta)\Bigg]\cr
=&~ \Bigg[(1+\alpha)4n^2\Delta^2(T_{k+1}-T_{k})\int_{T_{k}\Delta}^{T_{k+1}\Delta}a^2(s)\int_{n_{\Delta,s}\Delta}^{(n_{\Delta,s}+1)\Delta}a^2(\tau)d\tau ds \cr
&~\times\exp\Bigg(2n(T_{k+1}\Delta-T_{k}\Delta)
\int_{T_{k}\Delta}^{T_{k+1}\Delta}a^2(s)ds\Bigg)+1+\frac{1}{\alpha}\Bigg]\varpi(T_k\Delta)\Bigg[\mathbb{E}[V(T_k\Delta h)]\cr
&~+2\int_{\frac{T_k}{h}}^{\frac{T_{k+1}}{h}}\iota(sh\Delta)
\exp\Bigg(-\int_{s+1}^{\frac{T_k}{h}-4}\vartheta(\tau)d\tau\Bigg)ds
+C\iota(T_k\Delta-(h+1)\Delta)\Bigg]\cr
&~+(1+\alpha)4n(T_{k+1}\Delta-T_{k}\Delta)
\int_{T_{k}\Delta}^{T_{k+1}\Delta}a^2(s)\int_{n_{\Delta,s}\Delta}^{(n_{\Delta,s}+1)\Delta}d^2(\tau)d\tau ds\cr
&~\times\exp\left(2n(T_{k+1}\Delta-T_{k}\Delta)\int_{T_{k}\Delta}^{T_{k+1}\Delta}a^2(s)ds\right).
\end{align}
Take $\alpha=\frac{1}{\sqrt{\Delta}}$.} 
It follows from $\lim_{k\to\infty}D_k=0$ that
{\setlength\abovedisplayskip{6pt}
\setlength\belowdisplayskip{6pt}
\begin{align*}
&\lim\limits_{k\to\infty}\left(1+\frac{1}{\sqrt{\Delta}}\right)4n^2\Delta^2(T_{k+1}-T_{k})\int_{T_{k}\Delta}^{T_{k+1}\Delta}a^2(s) \int_{n_{\Delta,s}\Delta}^{(n_{\Delta,s}+1)\Delta}a^2(\tau)d\tau ds \\ &~~~~~~~~~~\times
\exp\Big(2n\Delta(T_{k+1}-T_{k})
\int_{T_{k}\Delta}^{T_{k+1}\Delta}a^2(s)ds\Big)=0,\\
&\lim\limits_{k\to\infty}\left(1+\frac{1}{\sqrt{\Delta}}\right) 4n\Delta(T_{k+1}-T_{k})
\int_{T_{k}\Delta}^{T_{k+1}\Delta}a^2(s)\int_{n_{\Delta,s}\Delta}^{(n_{\Delta,s}+1)\Delta}d^2(\tau)d\tau ds\\
&~~~~~~~~~~\times
\exp\Big(2n\Delta(T_{k+1}-T_{k})
\int_{T_{k}\Delta}^{T_{k+1}\Delta}a^2(s)ds\Big)=0,\\
&\lim\limits_{k\to\infty}\int_{\frac{T_{k}}{h}}^{\frac{T_{k+1}}{h}}\iota(rh\Delta) \exp\Bigg(-\int_{r+1}^{\frac{T_{k}}{h}-4}\vartheta(\tau)d\tau\Bigg)dr=0.
\end{align*}
From the definition of $\varpi(t)$ and $\iota(t)$ and Assumption \ref{as1}, we know that $\varpi(T_{k}\Delta)$ is bounded and $\lim_{k\to\infty}\iota(T_{k}
\Delta-(h+1)\Delta)=0$. From Lemma \ref{lemma6}, we know that $\lim_{k\to\infty}\mathbb{E}[V(T_{k}\Delta h)]=0$.
Then from (\ref{T1-4}), we have $\lim_{k\to\infty}\sup_{T_k\Delta\leq t\leq T_{k+1}\Delta}\mathbb{E}\Big[\|x(t)\|_2^2\Big]=0$, and thus $\lim_{t\to\infty}\mathbb{E}\Big[\|x(t)\|_2^2\Big]=0$.}
\qed
\vskip 0.2cm


In the following, we give the more intuitive conditions for the mean square asymptotic stability of the solution of (\ref{eq1}).
\begin{theorem}\label{theorem2}
\rm{For (\ref{eq1}), if there exists a constant $\Delta>0$, an integer $h>0$ and constants $0<\varepsilon_1\leq \varepsilon_2\leq\frac{1}{2}$, 
and the random time-varying coefficient matrices of (\ref{eq1}) satisfy \\
\rm{(a.1)} $\{A(t),t\geq 0\}$ is independent of $\{D(t),t\geq 0\}$,
$\|A(t)\|_2\leq a(t), \|D(t)\|_2\leq d(t)$~\rm{a.s.}, $\forall~t\geq 0$, where $a(t)=\mathcal O\bigg(\frac{1}{(t+1)^{\frac{1}{2}+\varepsilon_1}}\bigg), d(t)=\mathcal O\bigg(\frac{1}{(t+1)^{\frac{1}{2}+\varepsilon_2}}\bigg)$, \\
\rm{(a.2)} $\mu_2\left(\mathbb{E}\Big[\int_{mh\Delta}^{(m+1)h\Delta}A(s)ds\Big|\mathcal{F}(mh\Delta)\Big]\right)\leq -c(m)$~\text{a.s.}, $m=0,1,\cdots$, and
$\liminf_{m\to\infty} c(m)
(1+mh\Delta)^{\frac{1}{2}+\varepsilon_1}>0$,\\
then the solution of (\ref{eq1}) is mean square asymptotically stable.}
\end{theorem}

\proof
By $d(t)=\mathcal O\Big((t+1)^{-\frac{1}{2}-\varepsilon_2}\Big)$, we know that $d(t)$ monotonically decreases to zero and $d(t+\Delta)=\mathcal O(d(t))$.
From $\|A(t)\|_2\leq a(t)~\text{a.s.}$ and $a(t)=\mathcal O\Big((t+1)^{-\frac{1}{2}-\varepsilon_1}\Big)$, we have
{\setlength\abovedisplayskip{5pt}
\setlength\belowdisplayskip{5pt}
$$
\mathbb E\bigg[\Big(\max\limits_{k\Delta\leq s< (k+1)\Delta}\|A(s)\|_2\Big)^{2^{\max\{h,2\}}}\bigg|\mathcal F(k\Delta)\bigg]^{\frac{1}{2^{\max\{h,2\}}}}
\leq a(k\Delta).
$$
Denote $\rho(k)=a(k\Delta)$. It follows from $a(t)=\mathcal O\Big((t+1)^{-\frac{1}{2}-\varepsilon_1}\Big)$ that $\rho(k)=\mathcal O\Big((k\Delta+1)^{-\frac{1}{2}-\varepsilon_1}\Big)$.
Then from $d(t)=\mathcal O\Big((t+1)^{-\frac{1}{2}-\varepsilon_2}\Big)$,
$\liminf_{m\to\infty} c(m)(1+mh\Delta)^{\frac{1}{2}+\varepsilon_1}>0$ and $\varepsilon_2\geq \varepsilon_1$, we know that $\sum_{m=0}^{\infty}c(m)=\infty$, $d(k\Delta)=\mathcal O(\rho(k))$ and $\rho^2(kh)=o(c(k))$.
Denote $T_k=4h k^{2+\eta}$, where 
$\eta=\frac{4\varepsilon_1}{1-2\varepsilon_1}$.
From the conclusions (i)-(v) in Lemma \ref{lemma7}, it follows that $\lim_{k\to\infty}D_k=0$.
Then conditions in Theorem \ref{theorem1} are satisfied, so the solution of (\ref{eq1}) is mean square asymptotically stable.}
\qed

\vskip 0.2cm

\section{Asymptotic convergence of the algorithm}\label{sec:convergence}
Based on the results in Section \ref{sec:stability}, we study the asymptotic convergence of the algorithm (\ref{model3}).

\begin{theorem}\label{asycontheorem}
\rm{For (\ref{model3}), if there exist positive constants $\Delta$ and $\rho_0$, an integer $h>0$ and constants $0<\varepsilon_1\leq \varepsilon_2\leq\frac{1}{2}$, 
such that \\
\rm{(b.1)} 
$\alpha(t)=\mathcal O\Big(\frac{1}{(t+1)^{\frac{1}{2}+\varepsilon_1}}\Big)$,
$\beta(t)=\mathcal O\Big(\frac{1}{(t+1)^{\frac{1}{2}+\varepsilon_2}}\Big)$,\\
\rm{(b.2)} $\{\mathcal{H}(t),t\geq 0\}$ is independent of $\{w(t),t\geq 0\}$ and
$\|\mathcal{H}(t)\|_2\leq \rho_0$ a.s., \\
\rm{(b.3)} $\mu_2\left(-\mathbb{E}\Big[\int_{mh\Delta}^{(m+1)h\Delta}\alpha(s)\mathcal{H}^{\top}(s)\mathcal{H}(s)+\beta(s)(\mathcal{L_G}\otimes I_N)ds\Big|\mathcal{F}(mh\Delta)\Big]\right)\leq -c(m)$~\text{a.s.}, $m=0,1,\cdots$, and
$\liminf_{m\to\infty} c(m)(1+mh\Delta)^{\frac{1}{2}+\varepsilon_1}>0$,\\
then the algorithm (\ref{model3}) asymptotically converges in mean square.}
\end{theorem}

\proof
By Condition (b.2), we have
\begin{align*}
&\|A(t)\|_2=\|\alpha(t)\mathcal{H}^{\top}(t)\mathcal{H}(t)+\beta(t)(\mathcal{L}_{\mathcal{G}}\otimes I_n)\|_2\leq\alpha(t)\rho_0^2+\beta(t)\|\mathcal{L}_{\mathcal{G}}\|_2,\\
&\|D(t)\|_2=\|\beta(t)D\Sigma\|_2\leq\beta(t)\|D\Sigma\|_2.
\end{align*}
Denote $a(t)=\rho_0^2\alpha(t)+\|\mathcal{L}_{\mathcal{G}}\|_2\beta(t)$, $d(t)=\|D\Sigma\|_2\beta(t)$. Then by Condition (b.1), we have
$a(t)=\mathcal O\Big(\frac{1}{(t+1)^{\frac{1}{2}+\varepsilon_1}}\Big)$, $\beta(t)=\mathcal O\Big(\frac{1}{(t+1)^{\frac{1}{2}+\varepsilon_2}}\Big)$.
Then the conditions in Theorem \ref{theorem2} are satisfied, so the algorithm converges in mean square.
\qed
\vskip 0.2cm

{To further study the convergence of the algorithm (\ref{model3}) where the measurement matrices contain a Markov chain, i.e., $H_i(t)=p_i(t)H_i$, where $\{p_i(t),~t\geq0\},~i=1,2,\cdots N$ are independent Markov chains, and $H_i\in\mathbb{R}^{n_i\times N}(n_i\leq N)$.}
At first, we give the definition of the strongly 1-exponential ergodicity of a Markov chain.

\begin{definition}[{See \cite{Prieto}}]\label{exponentialergodic}
\rm{A Markov chain $\{{r(t)},\ t \geq 0\}$ on a countable state space $S$ with the transition function $(P_{ij}(t))_{i, j \in S}$  is strongly 1-exponential ergodic if there exists a distribution $\mu$ on $S$, constants $R>0$ and $\delta_{0}>0$ such that
{\setlength\abovedisplayskip{5pt}
\setlength\belowdisplayskip{1pt}
\begin{align*}
\sum_{j\in S}\left|P_{ij}(t)-\mu(j)\right| \leq R e^{-\delta_{0} t},  \quad \forall \  i \in S \text { and } t \geq 0.
\end{align*}}}
\end{definition}

In the following, we give conditions on the Markov chain, the graph and the algorithm gains.
Denote $r(t)=\mathcal{H}^{\top}(t)\mathcal{H}(t)$.

\begin{condition}\label{condition3}
\rm{\begin{itemize}
\item[\text{\rm{(i)}}]
$\{r(t),\mathcal{F}(t),t\geq 0\}$ is a matrix-valued Markov chain with a countable state space $E=\{r_j,j=0,1,2,\cdots\}$, and the generator $\Gamma=(\gamma_{ij})$ such that
$
P\{r(t+\delta)=r_{j} \mid r(t)=r_{i}\}= \gamma_{i j} \delta+\mathrm{o}(\delta), i \neq j,
\text {and}\  P\{r(t+\delta)=r_{j} \mid r(t)=r_{i}\}=1+\gamma_{i j} \delta+\mathrm{o}(\delta), i=j,
$
with $\delta>0$. Here, $\gamma_{ij}\geq 0$ is the transition rate from $r_{i}$ to $r_{j}$ if $i\neq j$ and $\gamma_{ii}=-\sum_{i\neq j,\ r_{j}\in E} \gamma_{ij}$;
\item[\text{\rm{(ii)}}]
there exists $\alpha_1>0$ such that 
$r_l\leq\alpha_1I_{Nn}$, $\forall~r_l\in E$, 
$\sup_{r_l\in E}|\gamma_{ii}|<\infty$, and $\{r(t), t\geq 0\}$ is independent of $\{w(t), t\geq 0\}$ and is strongly $1$-exponentially ergodic, with the unique stationary distribution $\pi= [\pi_1,\pi_2,\cdots]$, $\pi_j\geq0$, $ \sum_{j=0}^{\infty}\pi_{j}=1$,
and $\lambda_{\min}\left(\sum_{j=0}^{\infty} \pi_{j}r_{j}\right)>0$;
\item[\text{\rm{(iii)}}] there exists $\alpha_{2}>0$ such that
$(\mathcal{L}_{\mathcal G}+\mathcal{L}_{\mathcal G}^{\top})\otimes I_n\geq-\alpha_{2}I_{Nn}$;
\item[\text{\rm{(iv)}}]
there exist positive constants $0<\varepsilon_1\leq \varepsilon_2\leq\frac{1}{2}$, such that
$\alpha(t)=\mathcal O\Big(\frac{1}{(t+1)^{\frac{1}{2}+\varepsilon_1}}\Big)$,
$\beta(t)=\mathcal O\Big(\frac{1}{(t+1)^{\frac{1}{2}+\varepsilon_2}}\Big)$.
\end{itemize}}
\end{condition}

Based on Condition \ref{condition3}, we give the sufficient condition for mean square convergence of the algorithm (\ref{model3}).
The proof of Corollary \ref{coro2} is given in Appendix \ref{appendix:A}.


\begin{corollary}\label{coro2}
\rm{For the continuous-time algorithm (\ref{model3}), if Condition \ref{condition3} holds, 
then the algorithm converges in mean square.}
\end{corollary}

\section{A numerical example}
\label{sec:example}
{Consider a balanced fixed digraph consisting of $10$ nodes. Suppose that the relationships among nodes are described by the graph $\mathcal{G}=\{\mathcal{V}, \mathcal{E}_{\mathcal{G}}, \mathcal{A}_{\mathcal{G}}\}$,
where $\mathcal{V}=\{1,2,\cdots,10\}$ is the set of nodes, $\mathcal{E}_{\mathcal{G}}$ is the set of edges, and $\mathcal{A}_{\mathcal{G}}=[a_{ij}]_{i,j=1}^{10}$ is the adjacency matrix with  
$a_{13}=0.4,\ a_{15}=0.7,\ a_{18}=0.6,\ a_{21}=0.9,\ a_{28}=0.7,\ a_{29}=0.5,\ a_{34}=0.4,\ a_{37}=0.8,\ a_{39}=0.9,\ a_{41}=0.9,\ a_{46}=0.2,\ a_{47}=0.6,\ a_{52}=0.1,\ a_{53}=0.6,\ a_{57}=0.9,\ a_{68}=0.8,\ a_{69}=0.7,\ a_{6,10}=0.4,\ a_{73}=0.5,\ a_{75}=0.9,\ a_{7,10}=0.1,\ a_{85}=0.8,\ a_{86}=0.3,\ a_{8,10}=0.5,\ a_{92}=0.9,\ a_{94}=0.3,\ a_{96}=0.3,\ a_{10,1}=0.6,\ a_{10,2}=0.5,\ a_{10,4}=0.7$. 
The communication topology is shown in Figure \ref{fig1}. 
All nodes over the network cooperatively estimate the
unknown parameter vector ${\theta}=[6,5,4]^{\top}$  by information exchange among nodes.
For each node $i\in\mathcal{V}$, its measurement of {$\theta$} follows (\ref{model1}) with $H_i(t)=p_i(t)H_i$, where $H_i\in\mathbb{R}^{3\times 3}$ are given as
\begin{align*}
&H_{1}=H_{4}=H_{7}=H_{10}=\begin{bmatrix}
  0 & 1.5 & 0 \\
  1.5 & 0 & 0 \\
  0 & 0 & 1.5
\end{bmatrix},
H_{2}=H_{5}=H_{8}=\begin{bmatrix}
  0 & -0.5 & 0 \\
  0 & 0 & 0.5 \\
  0.5 & 0 & 0
\end{bmatrix},\\
&H_{3}=H_{6}=H_{9}=\begin{bmatrix}
  0 & 0 & 0.5 \\
  0.5 & 0 & 0 \\
  0 & 0.5 & 0
\end{bmatrix},
\end{align*}
and $\{p_i(t),~t\geq0\},~i=1,2,\cdots,10$ are independent Markov chains with the states $0$ and $1$,  
whose transition rate matrices are
$\Gamma_1=\Gamma_4=\Gamma_7=\Gamma_{10}=\begin{bmatrix}
  -1 & 1 \\
  1 & -1
\end{bmatrix}$,
$\Gamma_2=\Gamma_5=\Gamma_8=\begin{bmatrix}
  -1 & 1 \\
  2 & -2
\end{bmatrix}$,
$\Gamma_3=\Gamma_6=\Gamma_9=\begin{bmatrix}
  -1 & 1 \\
  3 & -3
\end{bmatrix}$,
respectively. 
The update of the estimate by each node follows (\ref{model3}),
where the initial values  
${\theta_1(0)}=[3.1,\,8.5,\,2.7]^{\top},\ 
{\theta_2(0)}=[1.7,\,2.5,\,3.3]^{\top},\ 
{\theta_3(0)}=[3.3,\,5.5,\,4.4]^{\top},\ 
{\theta_4(0)}=[6.8,\,2.4,\,7.4]^{\top},\ 
{\theta_5(0)}=[0.5,\,7.6,\,4.3]^{\top},\ 
{\theta_6(0)}=[5.5,\,-0.9,\,1.2]^{\top},\ \\
{\theta_7(0)}=[8.5,\,11.1,\,4.2]^{\top},\ 
{\theta_8(0)}=[6.8,\,10.3,\,10.3]^{\top},\ 
{\theta_9(0)}=[-0.6,\,1.8,\,3.7]^{\top},\ 
{\theta_{10}(0)}=[9.2,\,0.9,\,4.1]^{\top}.$

Denote $r(t)=\mathcal{H}^{\top}(t)\mathcal{H}(t)$. It can be seen that 
$r(t)$ is an irreducible Markov chain. By Theorem 2.22 in \cite{Prieto}, we know that 
$r(t)$ is strongly 1-exponential ergodic. 
Here, 
$r_j\leq 2.25I_{30}$, $j=1,2,\cdots,1024$,  
and $\lambda_{\min}(\sum_{j=1}^{1024}\pi_jr_j)>0$. 
From the setting of the graph, we have  $(\mathcal{L}_{\mathcal{G}}+\mathcal{L}_{\mathcal{G}}^{\top})\otimes I_3\geq -0.2I_{30}$ and $\|\mathcal{L}_{\mathcal{G}}+\mathcal{L}_{\mathcal{G}}^{\top}\|_2\leq 6.2$. 
Take $h=1,~\Delta=0.1$. 
{The settings of the noise intensity and algorithm gains are listed in Table \ref{tab:settings}. 
Then Conditions in Corollary \ref{coro2} hold.
It can be seen from Figure \ref{fig2} and Figure \ref{fig3} that the larger noise intensity leads to slower convergence rates of the mean square errors of the estimations. From Figure \ref{fig3} and Figure \ref{fig4}, it can be seen that larger algorithm gain $\alpha(t)$ leads to faster convergence rates. }}

\begin{table}[htbp]
  \centering
  \caption{{Settings of the noise intensity and the algorithm gains.}}
  \label{tab:settings}
  \begin{tabular}{lccc}
    \toprule
     & $\sigma_{ji}$ & $\alpha(t)$ & $\beta(t)$ \\
    \midrule
    Setting 1  & $[0.1,0.1,0.1]^{\top}$ & $(t+1)^{-0.6}$ & $(t+1)^{-0.7}$ \\
    Setting 2  & $[2,2,2]^{\top}$ & $(t+1)^{-0.6}$ & $(t+1)^{-0.7}$ \\
    Setting 3  & $[2,2,2]^{\top}$ & $(t+1)^{-0.51}$ & $(t+1)^{-0.7}$ \\
    \bottomrule
  \end{tabular}
\end{table}

\begin{figure}[htbp]
  \begin{minipage}[t]{0.5\linewidth}
    \centering
    \includegraphics[width=8cm, height=5cm]{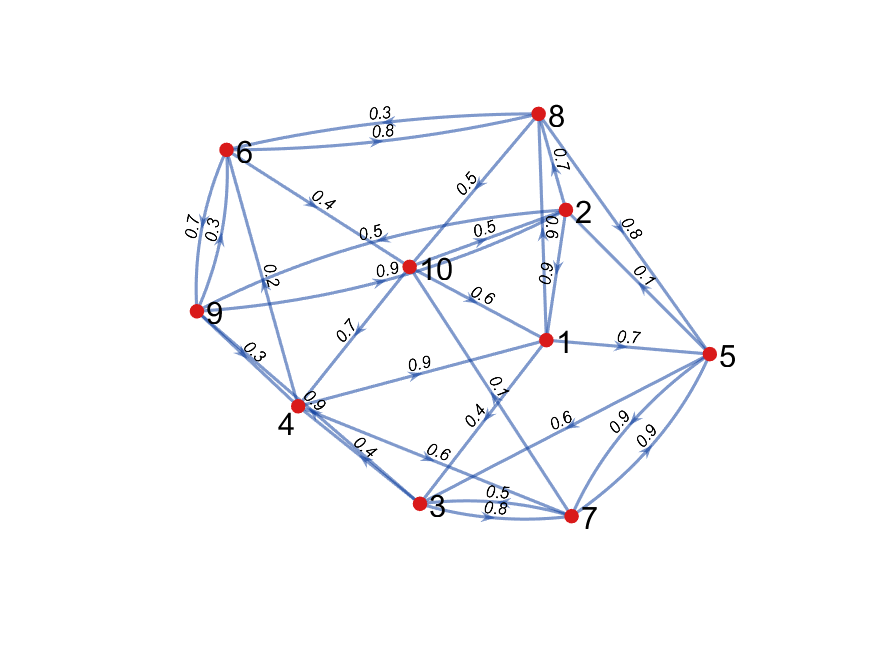}
    \caption{Communication topology.}
    \label{fig1}
  \end{minipage}%
  \hfill
  \begin{minipage}[t]{0.5\linewidth}
    \centering
    \includegraphics[width=8cm, height=5cm]{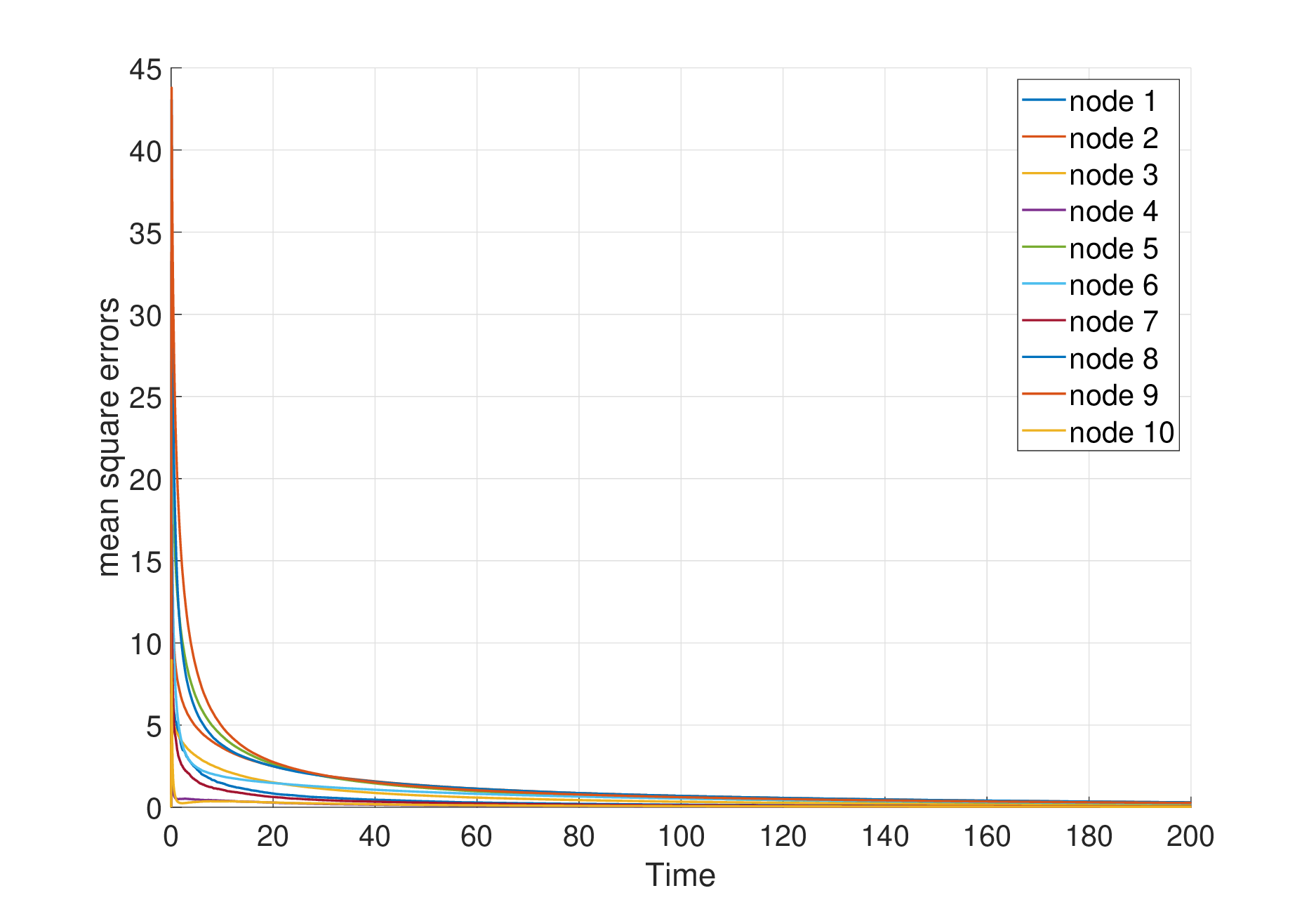}
    \caption{Trajectories of mean square errors \\for setting 1.}
    \label{fig2}
  \end{minipage}
  \par\vspace{1em} 
%
  \begin{minipage}[t]{0.5\linewidth}
    \centering
    \includegraphics[width=8cm, height=5cm]{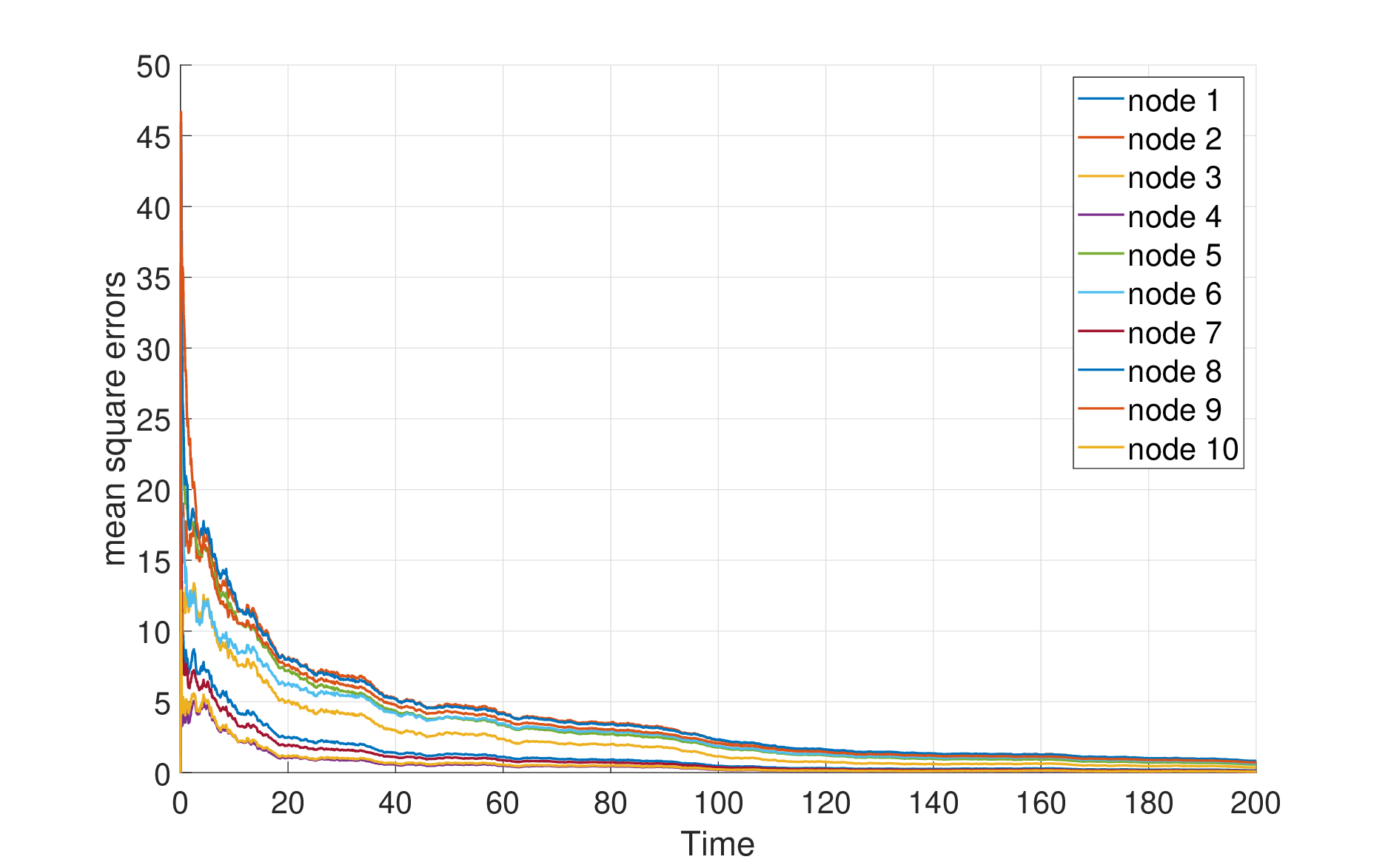}
    \caption{{Trajectories of mean square errors\\ for setting 2.}}
    \label{fig3}
  \end{minipage}%
  \hfill
  \begin{minipage}[t]{0.5\linewidth}
    \centering
    \includegraphics[width=8cm, height=5cm]{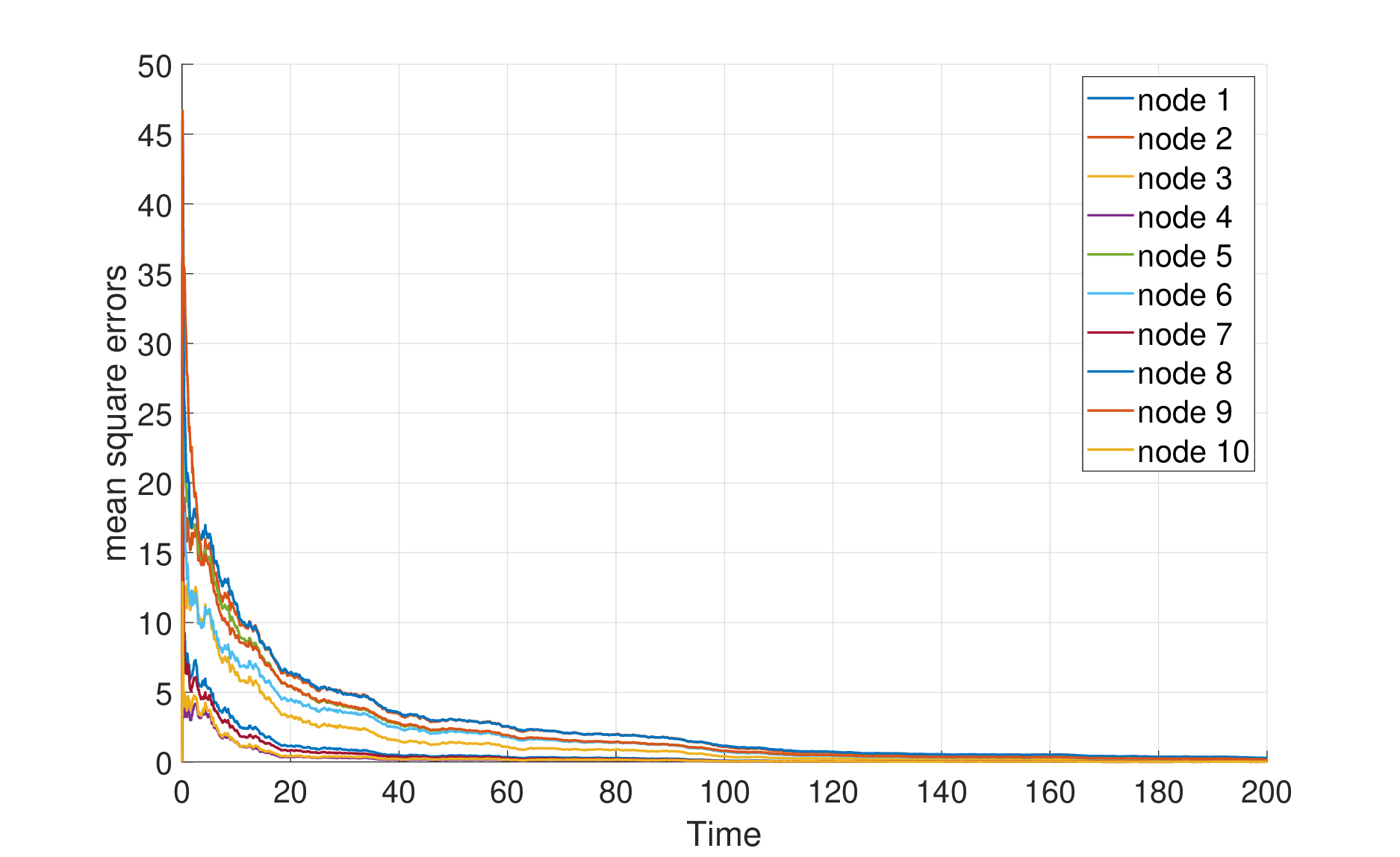}
    \caption{{Trajectories of mean square errors\\ for setting 3.}}
    \label{fig4}
  \end{minipage}
\end{figure}

\section{Conclusions}
\label{sec:conclusions}

In this work, convergence conditions have been examined for the continuous-time decentralized online estimation algorithm with the additive noises. 
By using stochastic differential equation theory and numerical approximation theory, we established  the stochastic stability of the SDEs with random time-varying coefficients.
It was proved that if the random measurement matrices and the graph satisfy some stochastic spatial-temporal persistence of excitation condition, then the algorithm gains can be properly designed to guarantee the mean square convergence.
Especially, it was proved that for a special case where the measurement matrices contain a Markov chain, the algorithms gains can be properly designed to guarantee the mean square convergence if the Markov chain is strongly $1$-exponentially ergodic. 

{In addition, some issues can be further investigated.
For the continuous-time decentralized online estimation with multiplicative noises, new techniques
need to be developed to separate the term coupling the random {measurement}  matrices and the multiplicative noise. 
To track the time-varying signals, new algorithms and techniques need to be developed to derive the bound of the tracking error.
A model-based algorithm is studied in this paper, and it is worthwhile to further study the model-free algorithm by the data-driven methods \cite{SPYX2024TAC,SWYX2025TAC}. 
It is also interesting to consider the model by using the reinforcement learning method in \cite{CWS2025SCIS}. }

\begin{appendices}

\section{Supplementary lemmas and proofs}
\label{appendix:A}
\def\thelemma{A.\arabic{lemma}}
\def\thecorollary{A.\arabic{corollary}}
\setcounter{equation}{0}
\def\theequation{A.\arabic{equation}}
\vskip 0.2cm



The proof of Lemma \ref{lemma5} needs the following Lemmas \ref{lemma1} and \ref{lemma4}.

\begin{lemma}[{See \cite{Guo2020time}}]\label{lemma1}
\rm{Suppose that $\{s_1(k),k\geq 0\}$ and $\{s_2(k),k\geq 0\}$ are sequences of real numbers, which satisfy $0\leq s_2(k)<1, \sum_{k=1}^{\infty}s_2(k)=\infty$, and $\lim_{k\to\infty}\frac{s_1(k)}{s_2(k)}$ exists, then
$\lim_{k\to\infty}\sum_{i=1}^ks_1(i)\prod_{l=i+1}^k(1-s_2(l))=\lim_{k\to\infty}\frac{s_1(k)}{s_2(k)}.$}
\end{lemma}
\vskip 0.2cm


Denote
\bna\label{BL5}
P(k+1)=I_n+\widetilde{A}(k+1).
\ena

\begin{lemma}\label{lemma4}
\rm{For (\ref{EM5}), 
if there exists an integer $h>0$, a positive real sequence $\{c(m),m\geq 0\}$ tending to zero 
and a positive real sequence $\{\rho(m),m\geq 0\}$ 
monotonically decreasing to zero,
such that \\
\rm{(i)}~ $\mu_2\left(\mathbb{E}\Big[\int_{mh\Delta}^{(m+1)h\Delta}A(s)ds\Big|\mathcal{F}(mh\Delta)\Big]\right) \leq -c(m)$~\rm{a.s.},$~m=0, 1, \cdots$,\\
\rm{(ii)}~ 
$\mathbb{E}\bigg[\Big(\max\limits_{k\Delta\leq s< (k+1)\Delta}\|A(s)\|_2\Big)^{2^{\max\{h,2\}}}\bigg|\mathcal{F}(k\Delta)\bigg]^{\frac{1}{2^{\max\{h,2\}}}}\leq \rho(k)$~\rm{a.s.},$~k=0,1, \cdots$,\\
then
\bna\label{lemma321}
&&\hspace{-1.8cm}\Big\|\mathbb{E}\Big[\Phi_P^{\top}((m+1)h,mh+1)\Phi_P((m+1)h,mh+1)\Big|\mathcal{F}(mh\Delta)\Big]\Big\|_2\cr
&&\hspace{-2.2cm}\leq 1-c(m)+\Big((1+\Delta)^{2h}-1-2h\Delta\Big)\rho^2(mh),~m\geq k_1~\rm{a.s.}
\ena
In particular, there exists a positive integer $k_2$ such that, for any $m\geq k_2$ and $mh+1\leq i\leq(m+1)h$,
\bna\label{lemma322}
\big\|\mathbb{E}\big[\Phi_P^{\top}((m+1)h, i+1)\Phi_P((m+1)h, i+1)\big|\mathcal{F}(i\Delta)\big]\big\|_2\leq 2~\text{a.s.}
\ena}
\end{lemma}
\noindent
\proof
From (\ref{BL5}), we have
{\setlength\abovedisplayskip{5pt}
\setlength\belowdisplayskip{5pt}
\ban
&&\hspace{-0.2cm}\Phi_P^{\top}((m+1)h,mh+1)\Phi_P((m+1)h,mh+1)\cr
&&\hspace{-0.6cm}=\Big(I_n+\widetilde{A}^{\top}(mh+1)\Big)\cdots\Big(I_n+\widetilde{A}^{\top}((m+1)h)\Big)\Big(I_n+\widetilde{A}((m+1)h)\Big)\cdots
\Big(I_n+\widetilde{A}(mh+1)\Big).
\ean
Then, we get}
{\setlength\abovedisplayskip{5pt}
\setlength\belowdisplayskip{5pt}
\begin{align}\label{L5-2}
&~\Big\|\mathbb{E}\Big[\Phi_P^{\top}((m+1)h,mh+1)\Phi_P((m+1)h,mh+1)\Big|\mathcal{F}(mh\Delta)\Big]\Big\|_2\cr
=&~\Big\|I_n+\sum\limits_{k=mh}^{(m+1)h-1}\mathbb{E}\Big[\widetilde{A}^{\top}(k+1)+\widetilde{A}(k+1)\Big|\mathcal{F}(mh\Delta)\Big]\cr
&~+\mathbb{E}[M_2(m)+\cdots+M_{2h}(m)\Big|\mathcal{F}(mh\Delta)]\Big\|_2\cr
\leq&~ \Big\|I_n+\sum\limits_{k=mh}^{(m+1)h-1}\mathbb{E}\Big[\widetilde{A}^{\top}(k+1)+\widetilde{A}(k+1)\Big|\mathcal{F}(mh\Delta)\Big]\Big\|_2\cr
&~+\|\mathbb{E}[M_2(m)+\cdots+M_{2h}(m)|\mathcal{F}(mh\Delta)]\|_2,
\end{align}
where $M_i(m), i=2,3,\cdots,2h$ denote the $i$-th order terms of the binomial expansion. For the first term on the r.h.s. of the above inequality, by the definition of the norm, we have}
{\setlength\abovedisplayskip{5pt}
\setlength\belowdisplayskip{5pt}
\bna\label{L5-3}
&&\hspace{-1.4cm}\left\|I_n+\sum\limits_{k=mh}^{(m+1)h-1}\mathbb{E}\Big[\widetilde{A}^{\top}(k+1)+\widetilde{A}(k+1)\Big|\mathcal{F}(mh\Delta)\Big]\right\|_2\cr
&&\hspace{-1.8cm}=\max\limits_{1\leq i\leq n}\left|1+\lambda_i\left(\sum\limits_{k=mh}^{(m+1)h-1}\mathbb{E}\Big[\widetilde{A}^{\top}(k+1)+\widetilde{A}(k+1)\Big|\mathcal{F}(mh\Delta)\Big]\right)\right|.
\ena
The conditional Lyapunov inequality and the condition (ii) lead to}
{\setlength\abovedisplayskip{5pt}
\setlength\belowdisplayskip{5pt}
\bna\label{L5-4-2}
&&\hspace{-1.4cm}\max\limits_{1\leq i\leq n}\lambda_i\left(\sum\limits_{k=mh}^{(m+1)h-1}\mathbb{E}\Big[-\widetilde{A}^{\top}(k+1)-\widetilde{A}(k+1)\Big|\mathcal{F}(mh\Delta)\Big]\right)\cr
&&\hspace{-1.8cm}\leq
\max\limits_{1\leq i\leq n}\left|\lambda_i\left(\sum\limits_{k=mh}^{(m+1)h-1}\mathbb{E}\Big[\widetilde{A}^{\top}(k+1)+\widetilde{A}(k+1)\Big|\mathcal{F}(mh\Delta)\Big]\right)\right|\cr
&&\hspace{-1.8cm}=
\left\|\sum\limits_{k=mh}^{(m+1)h-1}\mathbb{E}\Big[\widetilde{A}^{\top}(k+1)+\widetilde{A}(k+1)\Big|\mathcal{F}(mh\Delta)\Big]\right\|_2\cr
&&\hspace{-1.8cm}\leq 2 \sum\limits_{k=mh}^{(m+1)h-1}\mathbb{E}\left[\Big\|\widetilde{A}(k+1)\Big\|_2\Big|\mathcal{F}(mh\Delta)\right]
\leq 2h\Delta\rho(mh)~\text{a.s.}
\ena
As $\rho(k)$ monotonically decreases, there exists a positive integer $l_1$, such that} {\setlength\abovedisplayskip{5pt}
\setlength\belowdisplayskip{5pt}
\begin{align*}
&\lambda_i\left(\sum\limits_{k=mh}^{(m+1)h-1}\mathbb{E}\Big[- \widetilde{A}^{\top}(k+1)-\widetilde{A}(k+1)\Big|\mathcal{F}(mh\Delta)\Big]\right)\leq 1~\text{a.s.},
~i=1,2,\cdots,n,~m\geq l_1,
\end{align*} and then}
{\setlength\abovedisplayskip{5pt}
\setlength\belowdisplayskip{5pt}
\begin{align*}
&\lambda_i\left(\sum\limits_{k=mh}^{(m+1)h-1}\mathbb{E}\Big[\widetilde{A}^{\top}(k+1)+\widetilde{A}(k+1)\Big|\mathcal{F}(mh\Delta)\Big]\right)\geq -1~\text{a.s.},
~i=1,2,\cdots,n,~m\geq l_1.
\end{align*}
Then from (\ref{L5-3}), Lemma 5.4 in \cite{Xiong2008An} and the condition (i), we have for $\forall~m\geq l_1$,}
{\setlength\abovedisplayskip{5pt}
\setlength\belowdisplayskip{5pt}
\begin{align}\label{L5-6}
&\left\|I_n+\sum\limits_{k=mh}^{(m+1)h-1}\mathbb{E}\Big[\widetilde{A}^{\top}(k+1)+\widetilde{A}(k+1)\Big|\mathcal{F}(mh\Delta)\Big]\right\|_2\cr
=&1+2\mu_2\left(\mathbb{E}\bigg[\int_{mh\Delta}^{(m+1)h\Delta}A(s)ds\bigg|\mathcal{F}(mh\Delta)\bigg]\right)
\leq 1-c(m), 
~\text{a.s.}
\end{align}
By the conditional Lyapunov inequality and the condition (ii), we get}
{\setlength\abovedisplayskip{5pt}
\setlength\belowdisplayskip{5pt}
\ban
\mathbb{E}\Big[\big\|\widetilde{A}(k+1)\big\|_2^i\Big|\mathcal{F}(k\Delta)\Big]
\leq \mathbb{E}\Big[\big\|\widetilde{A}(k+1)\big\|_2^{2^h}\Big|\mathcal{F}(k\Delta)\Big]^{\frac{i}{2^h}}
\leq \Delta^i\rho^i(k)~\text{a.s.},~2\leq i\leq 2^h,
\ean
and from the monotonicity of $\rho(m)$, we get}
{\setlength\abovedisplayskip{5pt}
\setlength\belowdisplayskip{5pt}
\begin{align}\label{L5-7-2}
\mathbb{E}\bigg[\Big\|\widetilde{A}(k+1)\Big\|_2^i\bigg|\mathcal{F}(mh\Delta)\bigg]
=&~\mathbb{E}\bigg[\mathbb{E}\bigg[\Big\|\widetilde{A}(k+1)\Big\|_2^i\bigg|\mathcal{F}(k\Delta)\bigg]\bigg|\mathcal{F}(mh\Delta)\bigg]\cr
\leq&~ \Delta^i\rho^i(k)\leq \Delta^i\rho^i(mh)~\text{a.s.},~
k\geq mh,~2\leq i\leq2^h.
\end{align}
Notice that
$
\Phi_P^{\top}((m+1)h,mh+1)\Phi_P((m+1)h,mh+1)
= \sum_{i=0}^{2h}M_i(m),
$
where $M_i(m)=\sum_{p+q=i}\prod_{n=1}^p \widetilde{A}^{\top}(r_n)
\prod_{n=1}^q \widetilde{A}(s_{q+1-n}),~mh+1\leq r_1<\cdots<r_p\leq (m+1)h, mh+1\leq s_1<\cdots<s_q\leq(m+1)h$. From the conditional H\"{o}lder inequality and (\ref{L5-7-2}), it follows that}
{\setlength\abovedisplayskip{5pt}
\setlength\belowdisplayskip{5pt}
\begin{align*}
&~\mathbb{E}\left[\left\|\prod_{r=1}^j\widetilde{A}(s_r)\right\|_2\Bigg|\mathcal{F}(mh\Delta)\right]\cr
\leq&~ \mathbb{E}\left[\left\|\prod_{r=1}^{j-1}\widetilde{A}(s_r)\right\|_2\big\|\widetilde{A}(s_j)\big\|_2\Bigg|\mathcal{F}(mh\Delta)\right]\cr
\leq&~ \mathbb{E}\left[\left\|\prod_{r=1}^{j-1}\widetilde{A}(s_r)\right\|_2^2\Bigg|\mathbb{F}(mh\Delta)\right]^{\frac{1}{2}}\mathbb{E}\Bigg[\Big\|\widetilde{A}(s_j)\Big\|_2^2\Bigg|\mathcal{F}(mh\Delta)\Bigg]^{\frac{1}{2}}\cr
\leq&~ \mathbb{E}\left[\left\|\prod_{r=1}^{j-1}\widetilde{A}(s_r)\right\|_2^2\Bigg|\mathcal{F}(mh\Delta)\right]^{\frac{1}{2}}\Delta\rho(mh)
\leq \Delta^j\rho^j(mh).
\end{align*}
Then, we have}
{\setlength\abovedisplayskip{5pt}
\setlength\belowdisplayskip{5pt}
\begin{align*}
&~\mathbb{E}[\|M_i(m)\|_2|\mathcal{F}(mh\Delta)]\cr
=&~ \mathbb{E}\left[\left\|\sum\limits_{p+q=i}\prod_{n=1}^p\widetilde{A}^{\top}(r_n)\prod_{n=1}^q\widetilde{A}(s_{q+1-n})\right\|_2\Bigg|\mathcal{F}(mh\Delta)\right]\cr
\leq&~ \sum\limits_{p+q=i}\mathbb{E}\left[\left\|\prod_{n=1}^p\widetilde{A}^{\top}(r_n)\right\|_2\left\|\prod_{n=1}^q\widetilde{A}(s_{q+1-n})\right\|_2\Bigg|\mathcal{F}(mh\Delta)\right]\cr
\leq&~ \sum\limits_{p+q=i}\Delta^{p+q}\rho^{p+q}(mh)
=C_{2h}^i\Delta^i\rho^{i}(mh).
\end{align*}
As $\rho(m)$ is a positive sequence decaying to zero, there exists a positive constant $l_2$ such that $0<\rho(mh)\leq 1, m\geq l_2$, then for $2\leq i\leq 2h, m\geq l_2$, we have $\rho^i(mh)\leq\rho^2(mh)$. Thus, we have
\bna\label{L5-9}
\mathbb{E}[\|M_i(m)\|_2|\mathcal{F}(mh\Delta)]\leq \rho^2(mh)C_{2h}^i\Delta^i,~m\geq l_2,~2\leq i\leq 2h.
\ena
Then we get
\begin{align*}
\left\|\mathbb{E}\left[\sum\limits_{i=2}^{2h}M_i(m)\Bigg|\mathcal{F}(mh\Delta)\right]\right\|_2
\leq&~ \mathbb{E}\left[\left\|\sum\limits_{i=2}^{2h}M_i(m)\right\|_2\Bigg|\mathcal{F}(mh\Delta)\right]
\leq \sum\limits_{i=2}^{2h}\mathbb{E}[\|M_i(m)\|_2|\mathcal{F}(mh\Delta)]\|\cr
\leq&~ \rho^2 (mh)\sum\limits_{i=2}^{2h}C_{2h}^i\Delta^i
= \Big((1+\Delta)^{2h}-1-2h\Delta\Big)\rho^2(mh),\cr
&~~~~~~~~~~~~~~~~~~~~~~~~~~~~~~~~~~~~~~~~m\geq l_2,~2\leq i\leq 2h.
\end{align*}
Denote $k_1=\max\{l_1,l_2\}$. From (\ref{L5-6}) and the above inequality, we get (\ref{lemma321}).
Let $Q_j(m,i)$ be the $j$th order term in the binomial expansion of $\Phi_P^{\top}((m+1)h, i+1)\Phi_P((m+1)h, i+1)$. Similar to (\ref{L5-9}), we get}
$
\mathbb{E}[\|Q_j(m,i)\|_2|\mathcal{F}(i\Delta)]\leq C_{2mh+2h-2i}^j\rho^2(mh)\Delta^j~\text{a.s.}, mh+1\leq i\leq(m+1)h,~
 2\leq j\leq 2mh+2h-2i,~m\geq l_2.
$
Therefore, by (\ref{L5-4-2}), we have
{\setlength\abovedisplayskip{7pt}
\setlength\belowdisplayskip{7pt}
\ban
&&\hspace{-0.8cm}\Big\|\mathbb{E}\Big[\Phi_P^{\top}((m+1)h, i+1)\Phi_P((m+1)h, i+1)\Big|\mathcal{F}(i\Delta)\Big]\Big\|_2\cr
&&\hspace{-1.2cm}=\Bigg\|I_n+\sum\limits_{j=i}^{(m+1)h-1}\mathbb{E}\Big[\widetilde{A}^{\top}(j+1)+\widetilde{A}(j+1)\Big|\mathcal{F}(i\Delta)\Big]
+\mathbb{E}\Bigg[\sum\limits_{j=2}^{2mh+2h-2i}Q_j(m,i)\bigg|\mathcal{F}(i\Delta)\Bigg]\Bigg\|_2\cr
&&\hspace{-1.2cm}\leq
1+\left\|\sum\limits_{j=i}^{(m+1)h-1}\mathbb{E}\Big[\widetilde{A}^{\top}(j+1)+\widetilde{A}(j+1)\Big|\mathcal{F}(i\Delta)\Big]\right\|_2
+\left\|\mathbb{E}\left[\sum\limits_{j=2}^{2mh+2h-2i}Q_j(m,i)\bigg|\mathcal{F}(i\Delta)\right]\right\|_2\cr
&&\hspace{-1.2cm}\leq
1+\sum\limits_{j=mh}^{(m+1)h-1}\Big\|\mathbb{E}\Big[\widetilde{A}^{\top}(j+1)+\widetilde{A}(j+1)\Big|\mathcal{F}(i\Delta)\Big]\Big\|_2
+\left\|\mathbb{E}\left[\sum\limits_{j=2}^{2mh+2h-2i}Q_j(m,i)\bigg|\mathcal{F}(i\Delta)\right]\right\|_2\cr
&&\hspace{-1.2cm}\leq
1+2\Delta\sum\limits_{j=mh}^{(m+1)h-1}\rho(j)
+\sum\limits_{j=2}^{2mh+2h-2i}C_{2mh+2h-2i}^j\rho^2(mh)\Delta^j\cr
&&\hspace{-1.2cm}\leq
1+2\Delta\sum\limits_{j=mh}^{(m+1)h-1}\rho(j)+\Big((1+\Delta)^{2h}-1-2h\Delta\Big)\rho^2(mh)\cr
&&\hspace{-1.2cm}\leq 1+2h\Delta\rho(mh)+\Big((1+\Delta)^{2h}-1-2h\Delta\Big)\rho^2(mh)~\text{a.s.},~m\geq k_1.
\ean
As $\rho(\cdot)$ tends to zero monotonically, there exists a positive integer $k_2\geq k_1$ such that $1+2h\Delta\rho(mh)+\Big((1+\Delta)^{2h}-1-2h\Delta\Big)\rho^2(mh)\leq2$, that is,   (\ref{lemma322}) holds.}
\qed
\vskip 0.2cm

\noindent
{\bf Proof of Lemma \ref{lemma5}: }
From (\ref{EM5}), we have
{\setlength\abovedisplayskip{7pt}
\setlength\belowdisplayskip{7pt}
\begin{align*}
X\big((m+1)h\Delta\big)
=&~\Phi_P((m+1)h,mh+1)X(mh\Delta)
+\sum\limits_{i=mh+1}^{(m+1)h}\Phi_P((m+1)h,i+1)\xi(i),~m\geq 0.
\end{align*}
Then, we have}
{\setlength\abovedisplayskip{7pt}
\setlength\belowdisplayskip{7pt}
\bna\label{T2-2}
&&\hspace{-0.2cm}\mathbb{E}[V((m+1)h\Delta)|\mathcal{F}(mh\Delta)]\cr
&&\hspace{-0.6cm}= \mathbb{E}\Big[X^{\top}(mh\Delta)\Phi_P^{\top}((m+1)h,mh+1)\Phi_P((m+1)h,mh+1)
X(mh\Delta)\Big|\mathcal{F}(mh\Delta)\Big]\cr
&&\hspace{-0.2cm}+\mathbb{E}\Bigg[\sum\limits_{i=mh+1}^{(m+1)h}[\Phi_P((m+1)h,i+1)\xi(i)]^{\top}
\Phi_P((m+1)h,mh+1)
 X(mh\Delta)\Bigg|\mathcal{F}(mh\Delta)\Bigg]\cr
&&\hspace{-0.2cm}+\mathbb{E}\Bigg[[\Phi_P((m+1)h,mh+1)X(mh\Delta)]^{\top}\sum\limits_{i=mh+1}^{(m+1)h}\Phi_P((m+1)h,i+1)
\xi(i)\Bigg|\mathcal{F}(mh\Delta)\Bigg]\cr
&&\hspace{-0.2cm}+\mathbb{E}\Bigg[\left(\sum\limits_{i=mh+1}^{(m+1)h}\Phi_P((m+1)h,i+1)\xi(i)\right)^{\top}\cr
&&\hspace{0.2cm}\times\left(\sum\limits_{i=mh+1}^{(m+1)h}\Phi_P((m+1)h,i+1)\xi(i)\right)\Bigg|\mathcal{F}(mh\Delta)\Bigg].
\ena
For the first term on the r.h.s. of the above equality, from Lemma \ref{lemma4}, we have for $m\geq k_1$,}
{\setlength\abovedisplayskip{7pt}
\setlength\belowdisplayskip{7pt}
\bna\label{T2-6}
&&\hspace{-0.8cm}\mathbb{E}\Big[X^{\top}(mh\Delta)\Phi_P^{\top}((m+1)h,mh+1)
\Phi_P((m+1)h,mh+1)X(mh\Delta)\Big|\mathcal{F}(mh\Delta)\Big]\cr
&&\hspace{-1.2cm}= X^{\top}(mh\Delta)\mathbb{E}\Big[\Phi_P^{\top}((m+1)h,mh+1)
\Phi_P((m+1)h,mh+1)\Big|\mathcal{F}(mh\Delta)\Big]X(mh\Delta)\cr
&&\hspace{-1.2cm}\leq \Big\|\mathbb{E}\Big[\Phi_P^{\top}((m+1)h,mh+1)\Phi_P((m+1)h,mh+1)\Big|\mathcal{F}(mh\Delta)\Big]\Big\|_2
 X^{\top}(mh\Delta)X(mh\Delta)\cr
&&\hspace{-1.2cm}\leq \bigg(1-c(m)+\Big((1+\Delta)^{2h}-1-2h\Delta\Big)\rho^2(mh)\bigg)X^{\top}(mh\Delta)X(mh\Delta).
\ena
By Theorem 1.32 in \cite{Mao2006}, we have}
{\setlength\abovedisplayskip{7pt}
\setlength\belowdisplayskip{7pt}
\bna\label{T2-3-1}
\mathbb{E}[\xi(k+1)|\mathcal{F}(k\Delta)]=\mathbb{E}\left[\int_{k\Delta}^{(k+1)\Delta}D(s)dw(s)\Bigg|\mathcal{F}(k\Delta)\right]=0.
\ena
From Assumption \ref{as1}, 
we know that
$\tilde{A}(k)$ is independent with $\xi(k)$, and then  $\Phi_P^{\top}((m+1)h,i+1)\Phi_P((m+1)h,mh+1)$ is independent with $\xi(i)$. It follows from Lemma A.1 in \cite{li2018distributed} that  $\Phi_P^{\top}((m+1)h,i+1)\Phi_P((m+1)h,mh+1)$ and $\xi(i)$ are conditionally independent given $\mathcal{F}(mh\Delta)$. Then by (\ref{T2-3-1}), we have}
{\setlength\abovedisplayskip{7pt}
\setlength\belowdisplayskip{5pt}
\begin{align*}
&\mathbb{E}\Big[\xi^{\top}(i)\Phi_P^{\top}((m+1)h,i+1)\Phi_P((m+1)h,mh+1)\Big|\mathcal{F}(mh\Delta)\Big]X(mh\Delta)\cr
=&~\mathbb{E}\big[\xi^{\top}(i)\big|\mathcal{F}(mh\Delta)\big]\mathbb{E}\Big[\Phi_P^{\top}((m+1)h,i+1)
\Phi_P((m+1)h,mh+1)\Big|\mathcal{F}(mh\Delta)\Big]
 X(mh\Delta)\cr
=&~0.
\end{align*}
For the second term on the r.h.s. of (\ref{T2-2}), $X(mh\Delta)\in\mathcal{F}(mh\Delta)$ leads to}
{\setlength\abovedisplayskip{5pt}
\setlength\belowdisplayskip{5pt}
\begin{align}\label{T2-3}
&&\hspace{-0.2cm}\mathbb{E}\Bigg[\sum\limits_{i=mh+1}^{(m+1)h}[\Phi_P((m+1)h,i+1)\xi(i)]^{\top}
\Phi_P((m+1)h,mh+1)X(mh\Delta)\Bigg|\mathcal{F}(mh\Delta)\Bigg]\cr
&&\hspace{-0.6cm}=\sum\limits_{i=mh+1}^{(m+1)h}\mathbb{E}\Big[\xi^{\top}(i)
\Phi_P^{\top}((m+1)h,i+1)\Phi_P((m+1)h,mh+1)\Big|\mathcal{F}(mh\Delta)\Big]
 X(mh\Delta)=0.
\end{align}
For the third term on the r.h.s. of (\ref{T2-2}), similarly, we have}
{\setlength\abovedisplayskip{5pt}
\setlength\belowdisplayskip{5pt}
\bna\label{T2-3-3}
&&\hspace{-0.4cm}\mathbb{E}\Bigg[[\Phi_P((m+1)h,mh+1)X(mh\Delta)]^{\top}\cr
&&\hspace{-0.4cm}\times\sum\limits_{i=mh+1}^{(m+1)h}\Phi_P((m+1)h,i+1)\xi(i)\Bigg|\mathcal{F}(mh\Delta)\Bigg]=0.
\ena
From Assumption \ref{as1}, Lemma A.1 in \cite{li2018distributed} and (\ref{T2-3-1}), we have}
{\setlength\abovedisplayskip{3pt}
\setlength\belowdisplayskip{3pt}
\ban
&&\hspace{-0.2cm}\mathbb{E}\Big[\xi^{\top}(i)\Phi_P^{\top}((m+1)h,i+1)\Phi_P((m+1)h,j+1)\xi(j)\Big|\mathcal{F}(mh\Delta)\Big]\cr
&&\hspace{-0.6cm}=\mathbb{E}\Big[\mathbb{E}\big[\xi^{\top}(i)\Phi_P^{\top}((m+1)h,i+1)\Phi_P((m+1)h,j+1)\xi(j)\big|\mathcal{F}(j\Delta)\big]\Big|
\mathcal{F}(mh\Delta)\Big]\cr
&&\hspace{-0.6cm}=\mathbb{E}\Big[\mathbb{E}\big[\xi^{\top}(i)\big|\mathcal{F}(j\Delta)\big]\mathbb{E}\big[\Phi_P^{\top}((m+1)h,i+1)\Phi_P((m+1)h,j+1)\big|\mathcal{F}(j\Delta)\big]
\xi(j)\Big|\mathcal{F}(mh\Delta)\Big]\cr
&&\hspace{-0.6cm}=0,~mh+1\leq j<i\leq(m+1)h.
\ean
By {Assumption \ref{as1} and} Theorem 1.32 in \cite{Mao2006}, we get}
{\setlength\abovedisplayskip{3pt}
\setlength\belowdisplayskip{3pt}
\bna\label{T2-4-2}
&&\hspace{-0.8cm}\mathbb{E}\Big[\|\xi(k+1)\|_2^2|\mathcal{F}(k\Delta)\Big]
=\mathbb{E}\Bigg[\int_{k\Delta}^{(k+1)\Delta}\|D(s)\|_2^2ds\Bigg|\mathcal{F}(k\Delta)\Bigg]\cr
&&\hspace{3.4cm}\leq\int_{k\Delta}^{(k+1)\Delta}d^2(s)ds
\leq d^2(k\Delta)\Delta.
\ena
Therefore, for the fourth term on the r.h.s. of (\ref{T2-2}), from Assumption \ref{as1}, Lemma \ref{lemma4} and the above inequality, we have}
{\setlength\abovedisplayskip{3pt}
\setlength\belowdisplayskip{3pt}
\ban
&&\hspace{-0.2cm}\mathbb{E}\Bigg[\Bigg(\sum\limits_{i=mh+1}^{(m+1)h}
\Phi_P((m+1)h,i+1)\xi(i)\Bigg)^{\top}
\left(\sum\limits_{i=mh+1}^{(m+1)h}\Phi_P((m+1)h,i+1)\xi(i)\right)\Bigg|\mathcal{F}(mh\Delta)\Bigg]\cr
&&\hspace{-0.6cm}\leq \sum\limits_{i=mh+1}^{(m+1)h}\Big\|\mathbb{E}\Big[\xi^{\top}(i)\mathbb{E}\Big[\Phi_P^{\top}((m+1)h,i+1)
\Phi_P((m+1)h,i+1)\big|\mathcal{F}(i\Delta)\Big]
\xi(i)\Big|\mathcal{F}(mh\Delta)\Big]\Big\|_2\cr
&&\hspace{-0.6cm}\leq \sum\limits_{i=mh+1}^{(m+1)h}\mathbb{E}\bigg[\Big\|\mathbb{E}\Big[\Phi_P^{\top}((m+1)h,i+1)
\Phi_P((m+1)h,i+1)\Big|\mathcal{F}(i\Delta)\Big]\Big\|_2
\|\xi(i)\|_2^2\Big|\mathcal{F}(mh\Delta)\bigg]\cr
&&\hspace{-0.6cm}\leq2\sum\limits_{i=mh+1}^{(m+1)h}\mathbb{E}\Big[\mathbb{E}\Big[\|\xi(i)\|_2^2\Big|\mathcal{F}((i-1)\Delta)\Big]\Big|\mathcal{F}(mh\Delta)\Big]
\leq 2h\Delta d^2(mh\Delta),~m\geq k_2.
\ean
Substituting (\ref{T2-6}), (\ref{T2-3}), (\ref{T2-3-3}) and the above inequality into (\ref{T2-2}), we get}
{\setlength\abovedisplayskip{3pt}
\setlength\belowdisplayskip{3pt}
\bna\label{T2-7}
&&\hspace{-2.0cm}\mathbb{E}[V((m+1)h\Delta)|\mathcal{F}(mh\Delta)]
\leq \Big(1-c(m)+\big((1+\Delta)^{2h}-1-2h\Delta\big)\rho^2(mh)\Big)\cr
&&\hspace{3.4cm}\times X^{\top}(mh\Delta)X(mh\Delta)+ 2h\Delta d^2(mh\Delta),~m\geq k_2.
\ena
It follows from $\rho^2(mh)=o(c(m))$ that there exists a positive integer $k_3\geq k_2$, such that
$\big((1+\Delta)^{2h}-1-2h\Delta\big)\rho^2(mh)\leq \frac{1}{2}c(m),~m\geq k_3.$
Taking expectations on both sides of (\ref{T2-7}), we get}
{\setlength\abovedisplayskip{3pt}
\setlength\belowdisplayskip{3pt}
\ban
&&\hspace{-0.1cm}\mathbb{E}[V((m+1)h\Delta)]\leq\bigg(1-\frac{1}{2}c(m)\bigg)\mathbb{E}[V(mh\Delta)]+2h\Delta d^2(mh\Delta),~m\geq k_3.
\ean
Then, for any $L\geq0,~m\geq k_3$,} 
{\setlength\abovedisplayskip{3pt}
\setlength\belowdisplayskip{3pt}
\bna\label{T2-8-2}
&&\hspace{-1.2cm}\mathbb{E}[V((m+1)h\Delta)]\leq\prod_{i=L}^m\bigg(1-\frac{1}{2}c(i)\bigg)\mathbb{E}[V(Lh\Delta)]\cr
&&\hspace{2.4cm}+2h\Delta\sum_{i=L}^md^2(ih\Delta)\prod_{j=i+1}^m\bigg(1-\frac{1}{2}c(j)\bigg).
\ena
From (\ref{EM5}), we have
$
V\big((k+1)\Delta\big)= V(k\Delta)+X^{\top}(k\Delta)\widetilde{A}^{\top}(k+1)\widetilde{A}(k+1)X(k\Delta)+\xi^{\top}(k+1)\xi(k+1)
+2\xi^{\top}(k+1)\Big(I_n+\widetilde{A}(k+1)\Big)X(k\Delta)
+X^{\top}(k\Delta)\Big(\widetilde{A}(k+1)+\widetilde{A}^{\top}(k+1)\Big)X(k\Delta).
$
Denote $m_k=\lfloor\frac{k}{h}\rfloor$. By the above 
equality, we have} 
{\setlength\abovedisplayskip{3pt}
\setlength\belowdisplayskip{3pt}
\begin{align}\label{T2-9}
&~\mathbb{E}[V\big((k+1)\Delta\big)|\mathcal{F}(m_{k}h\Delta)]\cr
=&~ \mathbb{E}[V(k\Delta)|\mathcal{F}(m_{k}h\Delta)]
+\mathbb{E}\Big[X^{\top}(k\Delta)\widetilde{A}^{\top}(k+1)\widetilde{A}(k+1)
X(k\Delta)|\mathcal{F}(m_{k}h\Delta)\Big]\cr
&~+\mathbb{E}\Big[\xi^{\top}(k+1)\xi(k+1)|\mathcal{F}(m_{k}h\Delta)\Big]
+2\mathbb{E}\Big[\xi^{\top}(k+1)\Big(I_n+\widetilde{A}(k+1)\Big)
 X(k\Delta)|\mathcal{F}(m_{k}h\Delta)\Big]\cr
&~+\mathbb{E}\Big[X^{\top}(k\Delta)\Big(\widetilde{A}(k+1)
+\widetilde{A}^{\top}(k+1)\Big)
X(k\Delta)|\mathcal{F}(m_{k}h\Delta)\Big].
\end{align}
From $\mathcal{F}(m_{k}h\Delta)\subseteq\mathcal{F}(k\Delta)$,  $X(k\Delta)\in\mathcal{F}(k\Delta)$, Condition \ref{condition1} (ii), it follows that}
{\setlength\abovedisplayskip{3pt}
\setlength\belowdisplayskip{3pt}
\begin{align}\label{T2-10}
&~\mathbb{E}\Big[X^{\top}(k\Delta)\widetilde{A}^{\top}(k+1)\widetilde{A}(k+1)X(k\Delta)\Big|\mathcal{F}(m_{k}h\Delta)\Big]\cr
\leq&~ \mathbb{E}\Bigg[\Big\|\widetilde{A}(k+1)\Big\|_2^2\Big\|X(k\Delta)\Big\|_2^2\Bigg|\mathcal{F}(m_{k}h\Delta)\Bigg]\cr
=& ~ \mathbb{E}\Bigg[\mathbb{E}\Bigg[\Big\|\widetilde{A}(k+1)\Big\|_2^2\Bigg|\mathcal{F}(k\Delta)\Bigg]V(k\Delta)\Bigg|\mathcal{F}(m_{k}h\Delta)\Bigg]\cr
\leq&~ \Delta^2\rho^2(k)\mathbb{E}[V(k\Delta)|\mathcal{F}(m_{k}h\Delta)].
\end{align}
From $\mathcal{F}(m_{k}h\Delta)\subseteq\mathcal{F}(k\Delta)$ and (\ref{T2-4-2}), we get}
{\setlength\abovedisplayskip{3pt}
\setlength\belowdisplayskip{3pt}
\bna\label{T2-11}
&&\hspace{-0.8cm}\mathbb{E}\Big[\xi^{\top}(k+1)\xi(k+1)\Big|\mathcal{F}(m_{k}h\Delta)\Big]
\leq \mathbb{E}\Big[\mathbb{E}\Big[\|\xi(k+1)\|_2^2\Big|\mathcal{F}(k\Delta)\Big]\Big|\mathcal{F}(m_{k}h\Delta)\Big]\cr
&&\hspace{5.1cm}\leq \Delta d^2(k\Delta).
\ena
From $\mathcal{F}(m_{k}h\Delta)\subseteq\mathcal{F}(k\Delta)$,  $X(k\Delta)\in\mathcal{F}(k\Delta)$, Assumption \ref{as1}, Lemma A.1 in \cite{li2018distributed} and (\ref{T2-3-1}), we have}
{\setlength\abovedisplayskip{3pt}
\setlength\belowdisplayskip{3pt}
\bna\label{T2-12}
&&\hspace{-0.4cm}\mathbb{E}\Big[\xi^{\top}(k+1)\Big(I_n+\widetilde{A}(k+1)\Big)X(k\Delta)\Big|\mathcal{F}(m_{k}h\Delta)\Big]\cr
&&\hspace{-0.8cm}= \mathbb{E}\Big[\mathbb{E}\Big[\xi^{\top}(k+1)\Big(I_n+\widetilde{A}(k+1)\Big)X(k\Delta)\Big|\mathcal{F}(k\Delta)\Big]\Big|\mathcal{F}(m_{k}h\Delta)\Big]\cr
&&\hspace{-0.8cm}= \mathbb{E}\Big[\mathbb{E}\Big[\xi^{\top}(k+1)\Big|\mathcal{F}(k\Delta)\Big]\mathbb{E}\Big[\Big(I_n+\widetilde{A}(k+1)\Big)\Big|\mathcal{F}(k\Delta)\Big]
X(k\Delta)\Big|\mathcal{F}(m_{k}h\Delta)\Big]\cr
&&\hspace{-0.8cm}= 0.
\ena
From $\mathcal{F}(m_{k}h\Delta)\subseteq\mathcal{F}(k\Delta)$,  $X(k\Delta)\in\mathcal{F}(k\Delta)$, the conditional H$\rm{\ddot{o}}$lder inequality and Condition \ref{condition1} (ii), we have}
{\setlength\abovedisplayskip{3pt}
\setlength\belowdisplayskip{3pt}
\bna\label{T2-13}
&&\hspace{-0.8cm}\mathbb{E}\Big[X^{\top}(k\Delta)\Big(\widetilde{A}(k+1)+\widetilde{A}^{\top}(k+1)\Big)X(k\Delta)\Big|\mathcal{F}(m_{k}h\Delta)\Big]\cr
&&\hspace{-1.2cm}\leq \mathbb{E}\Big[\Big\|X^{\top}(k)\Big(\widetilde{A}(k+1)+\widetilde{A}^{\top}(k+1)\Big)X(k\Delta)\Big\|_2\Big|\mathcal{F}(m_{k}h\Delta)\Big]\cr
&&\hspace{-1.2cm}\leq 2\mathbb{E}\Big[\big\|\widetilde{A}(k+1)\big\|_2V(k\Delta)\Big|\mathcal{F}(m_{k}h\Delta)\Big]\cr
&&\hspace{-1.2cm}= 2\mathbb{E}\Big[\mathbb{E}\Big[\big\|\widetilde{A}(k+1)\big\|_2\Big|\mathcal{F}(k\Delta)\Big]V(k\Delta)\Big|\mathcal{F}(m_{k}h\Delta)\Big]\cr
&&\hspace{-1.2cm}\leq 2\mathbb{E}\Big[\mathbb{E}\Big[\big\|\widetilde{A}(k+1)\big\|_2^2\Big|\mathcal{F}(k\Delta)\Big]^{\frac{1}{2}}V(k\Delta)\Big|\mathcal{F}(m_{k}h\Delta)\Big]\cr
&&\hspace{-1.2cm}\leq 2\Delta\rho(k)\mathbb{E}[V(k\Delta)|\mathcal{F}(m_{k}h\Delta)].
\ena
By (\ref{T2-9})-(\ref{T2-13}), we have}
{\setlength\abovedisplayskip{5pt}
\setlength\belowdisplayskip{5pt}
\ban
&&\hspace{-0.4cm}\mathbb{E}[V\big((k+1)\Delta\big)|\mathcal{F}(m_{k}h\Delta)]
\leq \Delta d^2(k\Delta)+\big(1+2\Delta\rho(k)+\Delta^2\rho^2(k)\big) \mathbb{E}[V(k\Delta)|\mathcal{F}(m_{k}h\Delta)]
~\text{a.s.}
\ean
Then 
$
\mathbb{E}[V\big((k+1)\Delta\big)]
\leq \big(1+2\Delta\rho(k)+\Delta^2\rho^2(k)\big) \mathbb{E}[V(k\Delta)]
+\Delta d^2(k\Delta).
$
Iterating the above inequality leads to}
{\setlength\abovedisplayskip{5pt}
\setlength\belowdisplayskip{5pt}
\begin{align*}
\mathbb{E}[V\big((k+1)\Delta\big)]\leq& \prod_{i=m_kh}^k\Big(1+2\Delta\rho(i)+\Delta^2\rho^2(i)\Big)\mathbb{E}[V(m_kh\Delta)]\cr
&+\Delta\sum_{i=m_kh}^kd^2(i\Delta)\prod_{j=i+1}^k\Big(1+2\Delta\rho(j)+\Delta^2\rho^2(j)\Big).
\end{align*}
Combining the above inequality with (\ref{T2-8-2}), by the monotonicity of $\rho(\cdot)$ and $d(\cdot)$, we get
\begin{align*}
&\mathbb{E}[V\big((k+1)\Delta\big)]\cr
\leq& \prod_{i=m_kh}^k\Big(1+2\Delta\rho(i)+\Delta^2\rho^2(i)\Big)
\Bigg[\prod_{i=L}^{m_k-1}\left(1-\frac{1}{2}c(i)\right)\mathbb{E}[V(Lh\Delta)]\cr
&+2h\Delta\sum_{i=L}^{m_k-1}d^2(ih\Delta)\prod_{j=i+1}^{m_k-1}\left(1-\frac{1}{2}c(j)\right)\Bigg]
+\Delta\sum_{i=m_kh}^kd^2(i\Delta)
\prod_{j=i+1}^k\Big(1+2\Delta\rho(j)+\Delta^2\rho^2(j)\Big)\cr
\leq& \Big(1+2\Delta\rho(m_kh)+\Delta^2\rho^2(m_kh)\Big)^h\prod_{i=L}^{m_k-1}\left(1-\frac{1}{2}c(i)\right)\mathbb{E}[V(Lh\Delta)]\cr
&+2h\Delta\Big(1+2\Delta\rho(m_kh)+\Delta^2\rho^2(m_kh)\Big)^h\sum_{i=L}^{m_k-1}d^2(ih\Delta)\prod_{j=i+1}^{m_k-1}\left(1-\frac{1}{2}c(j)\right)\cr
&+h\Delta d^2(m_kh\Delta)\Big(1+2\Delta\rho(m_kh)+\Delta^2\rho^2(m_kh)\Big)^h,~L\geq0.
\end{align*}
Specially, for the inequality above, if $L=0$, then we have
\begin{align}\label{T2-17}
&~\mathbb{E}[V\big((k+1)\Delta\big)]\cr
\leq&~ \Big(1+2\Delta\rho(m_kh)+\Delta^2\rho^2(m_kh)\Big)^h\prod_{i=0}^{m_k-1}\bigg(1-\frac{1}{2}c(i)\bigg)\mathbb{E}[V(0)]\cr
&~+2h\Delta\Big(1+2\Delta\rho(m_kh)
+\Delta^2\rho^2(m_kh)\Big)^h
\sum_{i=0}^{m_k-1}d^2(ih\Delta)\prod_{j=i+1}^{m_k-1}\bigg(1-\frac{1}{2}c(j)\bigg)\cr
&~+h\Delta d^2(m_kh\Delta)\Big(1+2\Delta\rho(m_kh)+\Delta^2\rho^2(m_kh)\Big)^h.
\end{align}
For the first term on the r.h.s. of the inequality above, as $\rho(\cdot)$ tends to zero monotonically and $\sum_{m=0}^{\infty}c(m)=\infty$, we have
\bna\label{T2-17-1}
&&\hspace{-1.0cm}\lim_{k\to\infty}\Big(1+2\Delta\rho(m_kh)+\Delta^2\rho^2(m_kh)\Big)^h\prod_{i=0}^{m_k-1}\left(1-\frac{1}{2}c(i)\right)\mathbb{E}[V(0)]\cr
&&\hspace{-1.4cm}\leq\lim_{k\to\infty}\Big(1+2\Delta\rho(m_kh)+\Delta^2\rho^2(m_kh)\Big)^h\exp\left(-\frac{1}{2}\sum_{i=0}^{m_k-1}c(i)\right)\mathbb{E}[V(0)]\cr
&&\hspace{-1.4cm}=0.
\ena
For the second term on the r.h.s. of (\ref{T2-17}), as $c(m)$ tends to zero, we know that there exists $i_0$, such that if $i> i_0$, then $\frac{1}{2}c(i)\in(0,1)$. Then we have}
{\setlength\abovedisplayskip{7pt}
\setlength\belowdisplayskip{7pt}
\bna\label{T2-17-add}
&&\hspace{-1.2cm}\sum_{i=0}^{m_k-1}d^2(ih\Delta)\prod_{j=i+1}^{m_k-1}\left(1-\frac{1}{2}c(j)\right)\cr
&&\hspace{-1.6cm}=\sum_{i=0}^{i_0}d^2(ih\Delta)\prod_{j=i+1}^{m_k-1}\left(1-\frac{1}{2}c(j)\right)+\sum_{i=i_0+1}^{m_k-1}d^2(ih\Delta)\prod_{j=i+1}^{m_k-1}\left(1-\frac{1}{2}c(j)\right).
\ena
For the second term on the r.h.s. of the above equality, we have
\ban
&&\hspace{-0.4cm}\sum_{i=i_0+1}^{m_k-1}d^2(ih\Delta)\prod_{j=i+1}^{m_k-1}\left(1-\frac{1}{2}c(j)\right)
\leq\sum_{i=1}^{m_k-1}d^2(ih\Delta)\prod_{j=i+1}^{m_k-1}\left(1-s_2(j)\right),
\ean
where
{
\baselineskip=0pt
$
s_2(k)=\frac{1}{2}c(i_0+1),~k\leq i_0+1,$ and $s_2(k)=\frac{1}{2}c(k),~k> i_0+1.$
}
From $d(m\Delta)=\mathcal O(\rho(m))$ and $\rho^2(mh)=o(c(m))$, we get
\ban
&&\hspace{-0.4cm}\lim_{k\to\infty}\frac{d^2((m_k-1)h\Delta)}{\frac{1}{2}c(m_k-1)}
=\lim_{k\to\infty}\frac{d^2((m_k-1)h\Delta)}{\rho^2((m_k-1)h)}\frac{\rho^2((m_k-1)h)}{\frac{1}{2}c(m_k-1)}
=0,
\ean
then by Lemma \ref{lemma1},we have
\ban
&&\hspace{-0.8cm}\lim_{k\to\infty}\sum_{i=1}^{m_k-1}d^2(ih\Delta)\prod_{j=i+1}^{m_k-1}\left(1-s_2(j)\right)
=\lim_{k\to\infty}\frac{d^2((m_k-1)h\Delta)}{s_2(m_k-1)}
=\lim_{k\to\infty}\frac{d^2((m_k-1)h\Delta)}{\frac{1}{2}c(m_k-1)}
=0,
\ean
which gives
\bna\label{T2-17-add-4}
&&\hspace{-0.4cm}\lim_{k\to\infty}\sum_{i=i_0+1}^{m_k-1}d^2(ih\Delta)\prod_{j=i+1}^{m_k-1}\left(1-\frac{1}{2}c(j)\right)=0.
\ena
For the first term on the r.h.s. of (\ref{T2-17-add}), we have
\ban
&&\hspace{-0.2cm}\sum_{i=0}^{i_0}d^2(ih\Delta)\prod_{j=i+1}^{m_k-1}\left(1-\frac{1}{2}c(j)\right)
=\prod_{j=i_0+1}^{m_k-1}\left(1-\frac{1}{2}c(j)\right)\sum_{i=0}^{i_0}d^2(ih\Delta)\prod_{j=i+1}^{i_0}\left(1-\frac{1}{2}c(j)\right).
\ean
As $\prod_{j=i_0+1}^{m_k-1}\left(1-\frac{1}{2}c(j)\right)\leq\exp\left(-\sum_{i=i_0+1}^{m_k-1}\frac{1}{2}c(i)\right)$, and $\sum_{m=0}^\infty c(m)=\infty$, 
we know that
$$
\lim_{k\to\infty}\sum_{i=0}^{i_0}d^2(ih\Delta)\prod_{j=i+1}^{m_k-1}\left(1-\frac{1}{2}c(j)\right)=0.
$$
From (\ref{T2-17-add-4}) and the above equality, we have
\begin{align}\label{T2-17-2}
&&\hspace{-0.8cm}\lim_{k\to\infty}2h\Delta\Big(1+2\Delta\rho(m_kh)+\Delta^2\rho^2(m_kh)\Big)^h
\sum_{i=0}^{m_k-1}d^2(ih\Delta)\prod_{j=i+1}^{m_k-1}\left(1-\frac{1}{2}c(j)\right)
=0.
\end{align}
For the third term on the r.h.s. of (\ref{T2-17}), as $\rho(\cdot)$ and $d(\cdot)$ monotonically decrease to zero, we get}
{\setlength\abovedisplayskip{5pt}
\setlength\belowdisplayskip{5pt}
\ban
\lim_{k\to\infty}h\Delta d^2(m_kh\Delta)\Big(1+2\Delta\rho(m_kh)+\Delta^2\rho^2(m_kh)\Big)^h=0.
\ean
From (\ref{T2-17}), (\ref{T2-17-1}), (\ref{T2-17-2}) and the above equality, we get $\lim_{k\to\infty}\mathbb{E}[V(k\Delta)]=0$.}
\qed
\vskip 0.2cm

%

\noindent
{\bf Proof of Lemma \ref{lemma2}: }
Clearly, if the CTNAS is mean square asymptotically stable, then the DTNAS is mean square asymptotically stable.
Conversely, if the DTNAS is mean square asymptotically stable,
noting that $n_{\Delta,t}=\lfloor\frac{t}{\Delta}\rfloor$, then for any $t\geq n_{\Delta,T}\Delta$, we have $t\in[n_{\Delta,t}\Delta, (n_{\Delta,t}+1)\Delta)$ and $n_{\Delta,t}\Delta\geq n_{\Delta,T}\Delta$.
From (\ref{EM6}), we get for $t\geq n_{\Delta,T}\Delta$,
{\setlength\abovedisplayskip{8pt}
\setlength\belowdisplayskip{8pt}
\ban
&&\hspace{-0.45cm}\widehat{X}_{n_{\Delta,T}\Delta}(t)
=\widehat{X}_{n_{\Delta,T}\Delta}(n_{\Delta,t}\Delta)+\int_{n_{\Delta,t}\Delta}^tA(s)X(n_{\Delta,t}\Delta)ds+\int_{n_{\Delta,t}\Delta}^tD(s)dw(s).
\ean
Using the It$\rm{\hat{o}}$'s isometry, the H\"{o}lder inequality and Assumption \ref{as1}, we have}
{\setlength\abovedisplayskip{9pt}
\setlength\belowdisplayskip{9pt}
\begin{align*}
&~\mathbb{E}\bigg[\Big\|\widehat{X}_{n_{\Delta,T}\Delta}(t)\Big\|_2^2\bigg]\cr
\leq&~ 3\mathbb{E}\bigg[\Big\|\widehat{X}_{n_{\Delta,T}\Delta}(n_{\Delta,t}\Delta)\Big\|_2^2\bigg]+3\mathbb{E}\left[\Bigg\|\int_{n_{\Delta,t}\Delta}^tA(s)X(n_{\Delta,t}\Delta)ds\Bigg\|_2^2\right]
+3\mathbb{E}\left[\Bigg\|\int_{n_{\Delta,t}\Delta}^tD(s)dw(s)\Bigg\|_2^2\right]\cr
\leq&~ 3\mathbb{E}\bigg[\Big\|\widehat{X}_{n_{\Delta,T}\Delta}(n_{\Delta,t}\Delta)\Big\|_2^2\bigg]+3\mathbb{E}\Bigg[\Bigg(\int_{n_{\Delta,t}\Delta}^{t}\sqrt{n}\|A(s)\|_2ds\Bigg)^2\|X(n_{\Delta,t}\Delta)\|_2^2\Bigg]\cr
&~+3\Bigg[\int_{n_{\Delta,t}\Delta}^{t}\mathbb{E}\|D(s)\|_2^2 ds\Bigg]\cr
\leq&~ 3\mathbb{E}\bigg[\Big\|\widehat{X}_{n_{\Delta,T}\Delta}(n_{\Delta,t}\Delta)\Big\|_2^2\bigg]+3(t-n_{\Delta,t}\Delta)\mathbb{E}\Bigg[\int_{n_{\Delta,t}\Delta}^{t}n\|A(s)\|_2^2ds\|X(n_{\Delta,t}\Delta)\|_2^2\Bigg]\cr
&~+3\Bigg[\int_{n_{\Delta,t}\Delta}^{t}\mathbb{E}\|D(s)\|_2^2 ds\Bigg]\cr
\leq&~ 3\mathbb{E}\bigg[\Big\|\widehat{X}_{n_{\Delta,T}\Delta}(n_{\Delta,t}\Delta)\Big\|_2^2\bigg]
+3(t-n_{\Delta,t}\Delta)n\int_{n_{\Delta,t}\Delta}^ta^2(s)ds\mathbb{E}\big[\|X(n_{\Delta,t}\Delta)\|_2^2\big]\cr
&~+3\int_{n_{\Delta,t}\Delta}^{t}d^2(s)ds,~t\geq n_{\Delta,T}\Delta,
\end{align*}
which further gives}
{\setlength\abovedisplayskip{8pt}
\setlength\belowdisplayskip{8pt}
\bna\label{L2-4}
&&\hspace{-1.2cm}\sup\limits_{t\in[n_{\Delta,t}\Delta, (n_{\Delta,t}+1)\Delta)}\mathbb{E}\bigg[\Big\|\widehat{X}_{n_{\Delta,T}\Delta}(t)\Big\|_2^2\bigg]\cr
&&\hspace{-1.6cm}\leq 3\Bigg(1+n\Delta\int_{n_{\Delta,t}\Delta}^{(n_{\Delta,t}+1)\Delta}a^2(s)ds\Bigg)\mathbb{E}\Big[\|X(n_{\Delta,t}\Delta)\|_2^2\Big]
+3\int_{n_{\Delta,t}\Delta}^{(n_{\Delta,t}+1)\Delta}d^2(s)ds.
\ena
As $\int_{0}^\infty d^2(t)dt<\infty$, then $\sum_{k=0}^{\infty}\int_{k\Delta}^{(k+1)\Delta} d^2(t)dt<\infty$, we get $\lim_{k\to\infty}\int_{k\Delta}^{(k+1)\Delta}d^2(s)ds=0$, so  $\lim_{t\to\infty}\int_{n_{\Delta,t}\Delta}^{(n_{\Delta,t}+1)\Delta}d^2(s)ds=0$.
If $\lim_{t\to\infty}\mathbb{E}\Big[\|X(n_{\Delta,t}\Delta)\|_2^2\Big]=0$, then by (\ref{L2-4}), we have  $\lim_{t\to\infty}\mathbb{E}\Big[\|\widehat{X}_{n_{\Delta,T}\Delta}(t)\|_2^2\Big]=0$, i.e. the CTNAS is mean square asymptotically stable.}
\qed
\vskip 0.2cm

\noindent
{\bf Proof of Lemma \ref{lemma6}: }
By Lemma \ref{lemma5}, we know that the DTNAS (\ref{EM5}) is mean square asymptotically stable. Furthermore, it follows that the CTNAS (\ref{EM6}) is mean square asymptotically stable by Lemma \ref{lemma2}.
Let $L=n_{\Delta,T}$, from Lemma \ref{lemma5}, we have
{\setlength\abovedisplayskip{6pt}
\setlength\belowdisplayskip{6pt}
\bna\label{L6-add}
&&\hspace{-0.6cm}\mathbb{E}[V((k+1)\Delta)]\cr
&&\hspace{-1.0cm}\leq \Big(1+2\Delta\rho(m_kh)+\Delta^2\rho^2(m_kh)\Big)^h\exp\Bigg(-\frac{1}{2}\sum_{i=n_{\Delta,T}}^{m_k-1}c(i)\Bigg)\mathbb{E}[V(n_{\Delta,T}h\Delta)]\cr
&&\hspace{-0.6cm}+2h\Delta\Big(1+2\Delta\rho(m_kh)+\Delta^2\rho^2(m_kh)\Big)^h\sum_{i=n_{\Delta,T}}^{m_k-1}d^2(ih\Delta) \prod_{j=i+1}^{m_k-1}\Bigg(1-\frac{1}{2}c(j)\Bigg)\cr
&&\hspace{-0.6cm}+h\Delta d^2(m_kh\Delta)\Big(1+2\Delta\rho(m_kh)+\Delta^2\rho^2(m_kh)\Big)^h.
\ena
From (\ref{EM6}), we have}
{\setlength\abovedisplayskip{6pt}
\setlength\belowdisplayskip{6pt}
\ban
&&\hspace{-0.2cm}\widehat{X}_{n_{\Delta,T}\Delta}(t)=\widehat{X}_{n_{\Delta,T}\Delta}(k\Delta)+\int_{k\Delta}^tA(s)\overline{X}(s)ds+\int_{k\Delta}^tD(s)dw(s),
~k\Delta\leq t\leq (k+1)\Delta.
\ean
From the above equality, $\overline{X}(s)=\widehat{X}_{n_{\Delta,T}\Delta}(k\Delta),~s\in[k\Delta, (k+1)\Delta)$, the H\"{o}lder inequality, (\ref{L6-add}), Assumption \ref{as1} and the It$\rm{\hat{o}}$'s isometry, we get}
{\setlength\abovedisplayskip{5pt}
\setlength\belowdisplayskip{5pt}
\begin{align}\label{L6-4}
&~\mathbb{E}\Bigg[\Big\|\widehat{X}_{n_{\Delta,T}\Delta}(t)\Big\|_2^2\Bigg]\cr
\leq&~ 3\mathbb{E}\Bigg[\Big\|\widehat{X}_{n_{\Delta,T}\Delta}(k\Delta)\Big\|_2^2\Bigg]+3\mathbb{E}\left[\left\|\int_{k\Delta}^tA(s)\overline{X}(s)ds\right\|_2^2\right]+3\mathbb{E}\left[\left\|\int_{k\Delta}^tD(s)dw(s)\right\|_2^2\right]\cr
\leq&~ 3\mathbb{E}\Bigg[\Big\|\widehat{X}_{n_{\Delta,T}\Delta}(k\Delta)\Big\|_2^2\Bigg]
+3n\Delta\mathbb{E}\left[\int_{k\Delta}^t\|A(s)\|_2^2ds\Big\|\widehat{X}_{n_{\Delta,T}\Delta}(k\Delta)\Big\|_2^2\right]\cr
&~+3\mathbb{E}\left[\left\|\int_{k\Delta}^tD(s)dw(s)\right\|_2^2\right]\cr
=&~ 3\mathbb{E}\Bigg[\Big\|\widehat{X}_{n_{\Delta,T}\Delta}(k\Delta)\Big\|_2^2\Bigg]+3n\Delta\mathbb{E}\Bigg[\int_{k\Delta}^ta^2(s)ds\Big\|\widehat{X}_{n_{\Delta,T}\Delta}(k\Delta)\Big\|_2^2\Bigg]
+3\mathbb{E}\left[\left\|\int_{k\Delta}^tD(s)dw(s)\right\|_2^2\right]\cr
\leq&~ 3\Bigg(1+n\Delta\int_{k\Delta}^{(k+1)\Delta}a^2(s)ds\Bigg)\mathbb{E}\bigg[\Big\|\widehat{X}_{n_{\Delta,T}\Delta}(k\Delta)\Big\|_2^2\bigg]
+3\mathbb{E}\left[\left\|\int_{k\Delta}^tD(s)dw(s)\right\|_2^2\right]\cr
\leq&~ 3\Bigg(1+n\Delta\int_{k\Delta}^{(k+1)\Delta}a^2(s)ds\Bigg)\Big(1+2\Delta\rho(m_{k-1}h)+\Delta^2\rho^2(m_{k-1}h)\Big)^h\cr
&~\times\Bigg[\exp\Bigg(-\frac{1}{2}\sum_{i=n_{\Delta,T}}^{m_{k-1}-1}c(i)\Bigg)
\mathbb{E}[V(n_{\Delta,T} h\Delta)]+2h\Delta\sum_{i=n_{\Delta,T}}^{m_{k-1}-1}d^2(ih\Delta)
\prod_{j=i+1}^{m_{k-1}-1}\Bigg(1-\frac{1}{2}c(j)\Bigg)\cr
&~+h\Delta d^2(m_{k-1}h\Delta)\Bigg]+3\Delta d^2(k\Delta),
~t\in[k\Delta, (k+1)\Delta),~k\geq n_{\Delta,T}.
\end{align}
From $d(t+\Delta)=\mathcal{O}(d(t))$ and $k-m_{k-1}h\leq h+1$, we know that $d^2(k\Delta)=\mathcal O(d^2(m_{k-1}h\Delta))$, then there exists a constant $C$, such that $3\Delta d^2(k\Delta)\leq Ch\Delta d^2(m_{k-1}h\Delta),~k\geq n_{\Delta,T}$.
From (\ref{L6-4}), we further have}
{\setlength\abovedisplayskip{7pt}
\setlength\belowdisplayskip{3pt}
\ban
&&\hspace{-0.2cm}\mathbb{E}\left[\Big\|\widehat{X}_{n_{\Delta,T}\Delta}(t)\Big\|_2^2\right]\cr
&&\hspace{-0.6cm}\leq 3\Bigg(1+n\Delta\int_{k\Delta}^{(k+1)\Delta}a^2(s)ds\Bigg)\Big(1+2\Delta\rho(m_{k-1}h)+\Delta^2\rho^2(m_{k-1}h)\Big)^h
\Bigg[\exp\Bigg(-\frac{1}{2}\sum_{i=n_{\Delta,T}}^{m_{k-1}-1}c(i)\Bigg)\cr
&&\hspace{-0.2cm}\times\mathbb{E}[V(n_{\Delta,T} h\Delta)]+2h\Delta\sum_{i=n_{\Delta,T}}^{m_{k-1}-1}d^2(ih\Delta) 
\prod_{j=i+1}^{m_{k-1}-1}\Bigg(1-\frac{1}{2}c(j)\Bigg)
+Ch\Delta d^2(m_{k-1}h\Delta)\Bigg]\cr
&&\hspace{-0.6cm}\leq \varpi(t)\Bigg[
\mathbb{E}[V(n_{\Delta,T}h\Delta)]+2\int_{n_{\Delta,T}}^{\frac{k-1}{h}-1}\iota(sh\Delta)\exp\Bigg(-\int_{s+1}^{\frac{k-1}{h}-2}\vartheta(\tau)d\tau\Bigg)ds
+C\iota(t-(h+1)\Delta)\Bigg]\cr
&&\hspace{-0.6cm}\leq \varpi(t)\Bigg[
\mathbb{E}[V(n_{\Delta,T}h\Delta)]+2\int_{\frac{n_{\Delta,T}}{h}}^{\frac{t}{h\Delta}}\iota(sh\Delta)\exp\Bigg(-\int_{s+1}^{\frac{t}{h\Delta}-4}\vartheta(\tau)d\tau\Bigg)ds
+C\iota(t-(h+1)\Delta)\Bigg],\cr
&&\hspace{7.4cm} \ \ \  \ \  \ \ \ \ \ \  \ \ \ \ \ \
~t\in[k\Delta,(k+1)\Delta),~k\geq n_{\Delta,T}.
\ean
This completes the proof. }
\qed
\vskip 0.2cm

\noindent
{\bf Proof of Lemma \ref{lemma3}: }
For any $T\geq 0$, from (\ref{eq1}), we have
$x(t)=x(n_{\Delta,T}\Delta)+\int_{n_{\Delta,T}\Delta}^tA(s) x(s)ds+\int_{n_{\Delta,T}\Delta}^tD(s)dw(s), t\geq n_{\Delta,T}\Delta.$
From the above equality, (\ref{EM6}), the H\"{o}lder inequality and Assumption \ref{as1}, we have
{\setlength\abovedisplayskip{7pt}
\setlength\belowdisplayskip{7pt}
\bna\label{L3-2}
&&\hspace{0.4cm}\mathbb{E}\Big[\big\|\widehat{X}_{n_{\Delta,T}\Delta}(t)-x(t)\big\|_2^2\Big]\cr
&&\hspace{-0.0cm}=\mathbb{E}\bigg[\Big\|\int_{n_{\Delta,T}\Delta}^tA(s)\Big[\overline{X}(s)-\widehat{X}_{n_{\Delta,T}\Delta}(s)+\widehat{X}_{n_{\Delta,T}\Delta}(s)-x(s)]ds\Big\|_2^2\bigg]\cr
&&\hspace{-0.0cm}\leq\mathbb{E}\bigg[2\Big\|\int_{n_{\Delta,T}\Delta}^tA(s)\Big[\overline{X}(s)-\widehat{X}_{n_{\Delta,T}\Delta}(s)\Big]ds\Big\|_2^2
+2\Big\|\int_{n_{\Delta,T}\Delta}^tA(s)\Big[\widehat{X}_{n_{\Delta,T}\Delta}(s)-x(s)\Big]ds\Big\|_2^2\bigg]\cr
&&\hspace{-0.0cm}\leq 2n(t-n_{\Delta,T}\Delta)\mathbb{E}\bigg[\int_{n_{\Delta,T}\Delta}^t\|A(s)\|_2^2\big\|\overline{X}(s)-\widehat{X}_{n_{\Delta,T}\Delta}(s)\big\|_2^2ds\cr
&&\hspace{0.4cm}+\int_{n_{\Delta,T}\Delta}^t\|A(s)\|_2^2\big\|\widehat{X}_{n_{\Delta,T}\Delta}(s)-x(s)\big\|_2^2ds\bigg]\cr
&&\hspace{-0.0cm}\leq 2n(T'-n_{\Delta,T}\Delta)\mathbb{E}\bigg[\int_{n_{\Delta,T}\Delta}^{T'}a^2(s)\big\|\overline{X}(s)-\widehat{X}_{n_{\Delta,T}\Delta}(s)\big\|_2^2ds\bigg]\cr
&&\hspace{0.4cm}+2n(T'-n_{\Delta,T}\Delta)\mathbb{E}\bigg[\int_{n_{\Delta,T}\Delta}^{t}a^2(s)\big\|\widehat{X}_{n_{\Delta,T}\Delta}(s)-x(s)\big\|_2^2ds\bigg],
~n_{\Delta,T}\Delta\leq t\leq T'.
\ena
Denote $n_{\Delta,s}=\lfloor\frac{s}{\Delta}\rfloor$. If $s\geq n_{\Delta,T}\Delta$, we know that $n_{\Delta,s}\Delta=\lfloor \frac{s}{\Delta}\rfloor\Delta\geq n_{\Delta,T}\Delta$. Then we have}
{\setlength\abovedisplayskip{7pt}
\setlength\belowdisplayskip{7pt}
\bna\label{L3-4}
&&\hspace{-0.4cm}\widehat{X}_{n_{\Delta,T}\Delta}(s)=\widehat{X}_{n_{\Delta,T}\Delta}(n_{\Delta,s}\Delta)+\int_{n_{\Delta,s}\Delta}^sA(\tau)\overline{X}(\tau)d\tau+\int_{n_{\Delta,s}\Delta}^sD(\tau)dw(\tau),\cr
&&\hspace{7.7cm}
~s\in[n_{\Delta,s}\Delta, (n_{\Delta,s}+1)\Delta).
\ena
By the definition of $\overline{X}(s)$, we have}
{\setlength\abovedisplayskip{7pt}
\setlength\belowdisplayskip{5pt}
\ban
&&\hspace{-0.0cm}\mathbb{E}\bigg[\big\|\overline{X}(s)\big\|_2^2\bigg]\cr
&&\hspace{-0.4cm}= \mathbb{E}\bigg[\Big\|\sum\limits_{n_{\Delta,s}=0}^\infty X(n_{\Delta,s}\Delta) I_{[n_{\Delta,s}\Delta,(n_{\Delta,s}+1)\Delta)}(s)\Big\|_2^2\bigg]\cr
&&\hspace{-0.4cm}= \sum\limits_{n_{\Delta,s}=0}^\infty\mathbb{E} \bigg[\Big\|X(n_{\Delta,s}\Delta) I_{[n_{\Delta,s}\Delta,(n_{\Delta,s}+1)\Delta)}(s)\Big\|^{2}_2\bigg]\cr
&&\hspace{-0.0cm}+\sum\limits_{i\neq j}\mathbb{E}\Big[X^{\top}(n_{\Delta,i}\Delta)X(n_{\Delta,j}\Delta) I_{[n_{\Delta,i}\Delta,(n_{\Delta,i}+1)\Delta)}(s) I_{[n_{\Delta,j}\Delta,(n_{\Delta,j}+1)\Delta)}(s)\Big]\cr
&&\hspace{-0.4cm}= \sum\limits_{n_{\Delta,s}=0}^\infty\mathbb{E} \bigg[\Big\|X(n_{\Delta,s}\Delta) I_{[n_{\Delta,s}\Delta,(n_{\Delta,s}+1)\Delta)}(s)\Big\|^{2}_2\bigg],~s\in[n_{\Delta,s}\Delta, (n_{\Delta,s}+1)\Delta).
\ean
Therefore,}
{\setlength\abovedisplayskip{5pt}
\setlength\belowdisplayskip{5pt}
\ban
&&\hspace{-0.4cm}\mathbb{E}\bigg[\big\|\overline{X}(s)\big\|_2^2\bigg]\leq \sup\limits_{n_{\Delta,s}\Delta\leq r\leq s}\mathbb{E}\Big[\|X(r)\|_2^2\Big]\leq \sup\limits_{n_{\Delta,T}\Delta\leq t\leq T'}\mathbb{E}\Big[\|X(t)\|_2^2\Big].
\ean
Noting that $\overline{X}(s)=\widehat{X}_{n_{\Delta,T}\Delta}(n_{\Delta,s}\Delta),~s\in[n_{\Delta,s}\Delta,(n_{\Delta,s}+1)\Delta)$, by  (\ref{L3-4}), the H\"{o}lder inequality, the It$\rm{\hat{o}}$'s isometry, Assumption \ref{as1} and the above inequality, we get for any $s\in[n_{\Delta,s}\Delta, (n_{\Delta,s}+1)\Delta)$,}
{\setlength\abovedisplayskip{3pt}
\setlength\belowdisplayskip{3pt}
\ban
&&\hspace{0.1cm}\mathbb{E}\bigg[\Big\|\widehat{X}_{n_{\Delta,T}\Delta}(s)-\overline{X}(s)\Big\|_2^2\bigg]\cr
&&\hspace{-0.3cm}=\mathbb{E}\bigg[\bigg\|\int_{n_{\Delta,s}\Delta}^sA(\tau)\overline{X}(\tau)d\tau+\int_{n_{\Delta,s}\Delta}^sD(\tau)dw(\tau)\bigg\|_2^2\bigg]\cr
&&\hspace{-0.3cm}\leq 2\mathbb{E}\bigg[\bigg\|\int_{n_{\Delta,s}\Delta}^{(n_{\Delta,s}+1)\Delta}A(\tau)\overline{X}(\tau)d\tau\bigg\|_2^2\bigg]+2\mathbb{E}\bigg[\bigg\|\int_{n_{\Delta,s}\Delta}^{(n_{\Delta,s}+1)\Delta}D(\tau)dw(\tau)\bigg\|_2^2\bigg]\cr
&&\hspace{-0.3cm}\leq 2n\Delta\mathbb{E}\bigg[\int_{n_{\Delta,s}\Delta}^{(n_{\Delta,s}+1)\Delta}\|A(\tau)\|_2^2d\tau\Big\|\widehat{X}_{n_{\Delta,T}\Delta}(n_{\Delta,s}\Delta)\Big\|_2^2\bigg]
+2\int_{n_{\Delta,s}\Delta}^{(n_{\Delta,s}+1)\Delta}\mathbb{E}\big[\|D(\tau)\|_2^2\big]d\tau\cr
&&\hspace{-0.3cm}\leq 2n\Delta\int_{n_{\Delta,s}\Delta}^{(n_{\Delta,s}+1)\Delta}a^2(\tau)d\tau\mathbb{E}\bigg[\Big\|\widehat{X}_{n_{\Delta,T}\Delta}(n_{\Delta,s}\Delta)\Big\|_2^2\bigg]+2\int_{n_{\Delta,s}\Delta}^{(n_{\Delta,s}+1)\Delta}d^2(\tau)d\tau\cr
&&\hspace{-0.3cm}\leq 2n\Delta\int_{n_{\Delta,s}\Delta}^{(n_{\Delta,s}+1)\Delta}a^2(\tau)d\tau\sup\limits_{n_{\Delta,T}\Delta\leq r\leq T'}\mathbb{E}\bigg[\Big\|\widehat{X}_{n_{\Delta,T}\Delta}(r)\Big\|_2^2\bigg]
+2\int_{n_{\Delta,s}\Delta}^{(n_{\Delta,s}+1)\Delta}d^2(\tau)d\tau.
\ean
Substituting the above inequality into (\ref{L3-2}), we have for any $t\in[n_{\Delta,T}\Delta,T']$,}
{\setlength\abovedisplayskip{3pt}
\setlength\belowdisplayskip{3pt}
\ban
&&\hspace{-0.4cm}\mathbb{E}\bigg[\Big\|\widehat{X}_{n_{\Delta,T}\Delta}(t)-x(t)\Big\|_2^2\bigg]\cr
&&\hspace{-0.8cm}\leq 2n(T'-n_{\Delta,T}\Delta)\int_{n_{\Delta,T}\Delta}^{T'}a^2(s)\Bigg(2n\Delta\int_{n_{\Delta,s}\Delta}^{(n_{\Delta,s}+1)\Delta}a^2(\tau)d\tau
\sup\limits_{n_{\Delta,T}\Delta\leq r\leq T'}\mathbb{E}\bigg[\Big\|\widehat{X}_{n_{\Delta,T}\Delta}(r)\Big\|_2^2\bigg]\cr
&&\hspace{-0.4cm}+2\int_{n_{\Delta,s}\Delta}^{(n_{\Delta,s}+1)\Delta}d^2(\tau)d\tau\Bigg)ds
+2n(T'-n_{\Delta,T}\Delta)\mathbb{E}\bigg[\int_{n_{\Delta,T}\Delta}^{t}a^2(s)\Big\|\widehat{X}_{n_{\Delta,T}\Delta}(s)-x(s)\Big\|_2^2ds\bigg].
\ean
From the Gronwall inequality, we have}
{\setlength\abovedisplayskip{3pt}
\setlength\belowdisplayskip{3pt}
\ban
&&\hspace{0.2cm}\mathbb{E}\bigg[\Big\|\widehat{X}_{n_{\Delta,T}\Delta}(t)-x(t)\Big\|_2^2\bigg]\cr
&&\hspace{-0.2cm}\leq \bigg[4n(T'-n_{\Delta,T}\Delta)\int_{n_{\Delta,T}\Delta}^{T'}a^2(s)\bigg(n\Delta\int_{n_{\Delta,s}\Delta}^{(n_{\Delta,s}+1)\Delta}a^2(\tau)d\tau
\sup\limits_{n_{\Delta,T}\Delta\leq r\leq T'}\mathbb{E}\bigg[\Big\|\widehat{X}_{n_{\Delta,T}\Delta}(r)\Big\|_2^2\bigg]\cr
&&\hspace{0.2cm}+\int_{n_{\Delta,s}\Delta}^{(n_{\Delta,s}+1)\Delta}d^2(\tau)d\tau\bigg)ds\bigg]
\exp\bigg(2n(T'-n_{\Delta,T}\Delta)\int_{n_{\Delta,T}\Delta}^{t}a^2(s)ds\bigg),~t\in[n_{\Delta,T}\Delta,T'].
\ean
Then (\ref{L3-1}) is obtained by taking the supremum on both sides of the above inequality.}
\vskip 0.3cm

The proof of Theorem \ref{theorem2} needs the following Lemma \ref{lemma7}.

\begin{lemma}\label{lemma7}
\rm{Let $a(t)=\frac{a}{(t+1)^{\frac{1}{2}+\varepsilon_1}}$, $d(t)=\frac{d}{(t+1)^{\frac{1}{2}+\varepsilon_2}}$, $a>0$, $d>0$, $\Delta>0$, $\varepsilon_1\in\left(0,\frac{1}{2}\right)$, $\varepsilon_2\in\left(0,\frac{1}{2}\right)$, $T_k=4hk^{2+\eta}$, $\eta=\frac{4\varepsilon_1}{1-2\varepsilon_1}$, $h$ is a positive integer, 
$\vartheta(t)=\frac{1}{2}c(k)$, $t\in[k\Delta,(k+1)\Delta)$, where $c(k)=\frac{c}{(1+kh\Delta)^{\frac{1}{2}+\varepsilon_1}}$, and $\iota(t)=h\Delta d^2(k\Delta)$, $t\in[k\Delta,(k+1)\Delta)$, 
then we have the following conclusions:\\
$
\text{\rm{(i)}}~\lim\limits_{k\to\infty}(T_{k+1}\Delta-T_{k}\Delta)\int_{T_{k}\Delta}^{T_{k+1}\Delta}a^2(s)ds=\frac{4a^2(4h\Delta)^{1-2\varepsilon_1}}{(1-2\varepsilon_1)^2},\\
\text{\rm{(ii)}}~\lim\limits_{k\to\infty}(T_{k+1}\Delta-T_{k}\Delta)\int_{T_{k}\Delta}^{T_{k+1}\Delta}a^2(s)\int_{n_{\Delta,s}\Delta}^{(n_{\Delta,s}+1)\Delta}d^2(\tau)d\tau ds= 0,\\
\text{\rm{(iii)}}~\lim\limits_{k\to\infty}(T_{k+1}\Delta-T_{k}\Delta)\int_{T_{k}\Delta}^{T_{k+1}\Delta}a^2(s)\int_{n_{\Delta,s}\Delta}^{(n_{\Delta,s}+1)\Delta}a^2(\tau)d\tau ds= 0,\\
\text{\rm{(iv)}}~\lim\limits_{k\to\infty}\int_{\frac{T_{k}}{h}}^{\frac{T_{k+1}}{h}}\vartheta(\tau)d\tau=\frac{2c}{(1-2\varepsilon_1)4^{\varepsilon_1}h ^{\frac{1}{2}+\varepsilon_1}},\\
\text{\rm{(v)}}~\lim\limits_{k\to\infty}\int_{\frac{T_{k}}{h}}^{\frac{T_{k+1}}{h}}\exp\left(\int_{\frac{T_{k}}{h}-4}^{s+1}\vartheta(\tau)d\tau\right)\iota(sh\Delta)ds=0.
$}
\end{lemma}

\noindent
\proof
(i) From the definitions of $a(t)$ and $T_k$, we have
{\setlength\abovedisplayskip{3pt}
\setlength\belowdisplayskip{3pt}
\ban
&&\hspace{-0.0cm}\lim\limits_{k\to\infty}(T_{k+1}\Delta-T_{k}\Delta)\int_{T_{k}\Delta}^{T_{k+1}\Delta}a^2(s)ds\cr
&&\hspace{-0.4cm}=\lim\limits_{k\to\infty}-2h\Delta a^2\varepsilon_1^{-1}\Big[(k+1)^{2+\eta}-k^{2+\eta}\Big]\Big[\big(1+4h(k+1)^{2+\eta}\Delta\big)^{-2\varepsilon_1}
-\big(1+4hk^{2+\eta}\Delta\big)^{-2\varepsilon_1}\Big].
\ean
Denote $g_1(k)=k^{2+\eta},~ g_2(k)=\left(1+4hk^{2+\eta}\Delta\right)^{-2\varepsilon_1}$. By the differential mean-value theorem, there exist $\theta_1(k)\in(0,1),~\theta_2(k)\in(0,1)$, such that} 
{\setlength\abovedisplayskip{3pt}
\setlength\belowdisplayskip{3pt}
\begin{align}\label{L4-3}
&g_1(k+1)-g_1(k)
=(2+\eta)(k+\theta_1(k))^{1+\eta},\cr
&g_2(k+1)-g_2(k)
=\frac{-8\varepsilon_1h\Delta(2+\eta)(k+\theta_2(k))^{1+\eta}}{\left(1+4h\Delta(k+\theta_2(k))^{2+\eta}\right)^{1+2\varepsilon_1}},
\end{align}
which further gives}
\begin{align*}
&\lim\limits_{k\to\infty}(T_{k+1}\Delta-T_{k}\Delta)\int_{T_{k}\Delta}^{T_{k+1}\Delta}a^2(s)ds\cr
=&\lim\limits_{k\to\infty}\frac{16h^2a^2\Delta^2(2+\eta)^2{(k+\theta_2(k))}^{2+2\eta}}{\left(1+4h\Delta{(k+\theta_2(k))}^{2+\eta}\right)^{1+2\varepsilon_1}}
\left(\frac{k+\theta_1(k)}{k+\theta_2(k)}\right)^{1+\eta}.
\end{align*}
As $\theta_1(k)\in(0,1),~\theta_2(k)\in(0,1)$, we have
{\setlength\abovedisplayskip{3pt}
\setlength\belowdisplayskip{3pt}
\ban
\lim\limits_{k\to\infty}\left(\frac{k+\theta_1(k)}{k+\theta_2(k)}\right)^{1+\eta}
=\lim\limits_{k\to\infty}\left(\frac{1+\frac{\theta_1(k)}{k}}{1+\frac{\theta_2(k)}{k}}\right)^{1+\eta}
=1.
\ean
By $\eta=\frac{4\varepsilon_1}{1-2\varepsilon_1}$, we have $2+2\eta=(2+\eta)(1+2\varepsilon_1)$, which leads to
$$
\lim\limits_{k\to\infty}\frac{16h^2\Delta^2{(k+\theta_2(k))}^{2+2\eta}}{\left(1+4h\Delta{(k+\theta_2(k))}^{2+\eta}\right)^{1+2\varepsilon_1}}=(4h\Delta)^{1-2\varepsilon_1}.
$$
Therefore, we have (i).}
\vskip 0.2cm

(ii) From $n_{\Delta,s}=\lfloor\frac{s}{\Delta}\rfloor$, we have $\frac{s}{\Delta}-1\leq n_{\Delta,s}\leq \frac{s}{\Delta}$, which gives $s-\Delta\leq n_{\Delta,s}\Delta\leq s$.
If $\tau\in[n_{\Delta,s}\Delta, (n_{\Delta,s}+1)\Delta]$, we get $d(\tau)\leq d(n_{\Delta,s}\Delta)\leq d(s-\Delta)$. Thus, we have
$
\int_{n_{\Delta,s}\Delta}^{(n_{\Delta,s}+1)\Delta}d^2(\tau)d\tau
\leq d^2(n_{\Delta,s}\Delta)\Delta\leq d^2(s-\Delta)\Delta.
$
Denote $a(t)=a,~ d(t)=d$, $t\in[-\Delta, 0)$. We further have
{\setlength\abovedisplayskip{3pt}
\setlength\belowdisplayskip{3pt}
\begin{align}\label{L4-7}
&\lim\limits_{k\to\infty}(T_{k+1}\Delta-T_{k}\Delta)\int_{T_{k}\Delta}^{T_{k+1}\Delta}a^2(s)\int_{n_{\Delta,s}\Delta}^{(n_{\Delta,s}+1)\Delta}d^2(\tau)d\tau ds\notag\\
\leq & \lim\limits_{k\to\infty}(T_{k+1}\Delta-T_{k}\Delta)\int_{T_{k}\Delta}^{T_{k+1}\Delta}a^2(s-\Delta)d^2(s-\Delta)\Delta ds\notag\\
=&\lim\limits_{k\to\infty}4h\Delta\left[(k+1)^{2+\eta}-k^{2+\eta}\right]\left(-\frac{a^2d^2\Delta}{1+2\varepsilon_1+2\varepsilon_2}\right)
\bigg[\Big(1+4h(k+1)^{2+\eta}\Delta
-\Delta\Big)^{-(1+2\varepsilon_1+2\varepsilon_2)}\notag\\
&-\Big(1+4hk^{2+\eta}\Delta-\Delta\Big)^{-(1+2\varepsilon_1+2\varepsilon_2)}\bigg].
\end{align}
Denote $g_3(k)=\left(1+4h\Delta k^{2+\eta}-\Delta\right)^{-1-2\varepsilon_1-2\varepsilon_2}$. There exists $\theta_3(k)\in(0,1)$, such that 
$g_3(k+1)-g_3(k)
=-(1+2(\varepsilon_1+\varepsilon_2))(1+4h\Delta(k+\theta_3(k))^{2+\eta}-\Delta)^{-(2+2\varepsilon_1+2\varepsilon_2)}
4h\Delta(2+\eta) (k+\theta_3(k))^{1+\eta}.$
Substituting (\ref{L4-3}) and the above equality into (\ref{L4-7}), we obtain}
{\setlength\abovedisplayskip{3pt}
\setlength\belowdisplayskip{3pt}
\begin{align*}
&\lim\limits_{k\to\infty}(T_{k+1}\Delta-T_{k}\Delta)\int_{T_{k}\Delta}^{T_{k+1}\Delta}a^2(s)\int_{n_{\Delta,s}\Delta}^{(n_{\Delta,s}+1)\Delta}d^2(\tau)d\tau ds\\
\leq&\lim\limits_{k\to\infty}\frac{16h^2\Delta^3a^2d^2(2+\eta)^2{(k+\theta_1(k))}^{1+\eta}{(k+\theta_3(k))}^{1+\eta}}{\left(1+4h\Delta{(k+\theta_3(k))}^{2+\eta}-\Delta\right)^{2+2\varepsilon_1+2\varepsilon_2}}\\
=&\lim\limits_{k\to\infty}\frac{16h^2\Delta ^3a^2d^2(2+\eta)^2{(k+\theta_3(k))}^{2+2\eta}}{\left(1+{4h\Delta(k+\theta_3(k))}^{2+\eta}-\Delta\right)^{2+2\varepsilon_1+2\varepsilon_2}}.
\end{align*}
By $\eta=\frac{4\varepsilon_1}{1-2\varepsilon_1}$, 
we have
$2+2\eta-(2+\eta)(2+2\varepsilon_1+2\varepsilon_2)=-2-\eta-4\varepsilon_2-2\eta\varepsilon_2<0,$
which further gives}
{\setlength\abovedisplayskip{3pt}
\setlength\belowdisplayskip{3pt}
\ban
\lim\limits_{k\to\infty}\frac{16h^2\Delta^2{(k+\theta_3(k))}^{2+2\eta}}{\left(1+4h\Delta{(k+\theta_3(k))}^{2+\eta}-\Delta\right)^{2+2\varepsilon_1+2\varepsilon_2}}=0.
\ean
Thus,
$\lim_{k\to\infty}(T_{k+1}\Delta-T_{k}\Delta)\int_{T_{k}\Delta}^{T_{k+1}\Delta}a^2(s)\int_{n_{\Delta,s}\Delta}^{(n_{\Delta,s}+1)\Delta}d^2(\tau)d\tau ds\leq 0.$
As  $(T_{k+1}\Delta-T_{k}\Delta)\int_{T_{k}\Delta}^{T_{k+1}\Delta}\\a^2(s) \int_{n_{\Delta,s}\Delta}^{(n_{\Delta,s}+1)\Delta}d^2(\tau)d\tau ds\geq 0$,
we have (ii).}
\vskip 0.2cm

(iii) Take $d=a,~ \varepsilon_2=\varepsilon_1$ in (ii), we directly get the conclusion.
\vskip 0.2cm

(iv) 
From the definitions of $\vartheta(t)$ and $c(k)$, we know that
$
\vartheta(t)
=\frac{1}{2}c\Big(1+\lfloor\frac{t}{\Delta}\rfloor h\Delta\Big)^{-\frac{1}{2}-\varepsilon_1},~t\in[k\Delta,(k+1)\Delta).
$
On one hand, from the monotonicity of $\vartheta(t)$, we have
{\setlength\abovedisplayskip{3pt}
\setlength\belowdisplayskip{3pt}
\begin{align*}
\int_{\frac{T_{k}}{h}}^{\frac{T_{k+1}}{h}}\vartheta(t)dt
\leq&~\frac{c}{2\Big(1+\lfloor\frac{T_k}{h\Delta}\rfloor h\Delta \Big)^{\frac{1}{2}+\varepsilon_1}}\left(4(k+1)^{2+\eta}-4k^{2+\eta}\right)\cr
\leq&~\frac{c}{2\left(1+\left(\frac{4k^{2+\eta}}{\Delta}-1\right)h\Delta\right)^{\frac{1}{2}+\varepsilon_1}}\left(4(k+1)^{2+\eta}-4 k^{2+\eta}\right)\cr
\leq&~\frac{2 c(2+\eta)(k+1)^{1+\eta}}{\Big(1+4h k^{2+\eta}-h\Delta\Big)^{\frac{1}{2}+\varepsilon_1}}.
\end{align*}
From $\eta=\frac{4\varepsilon_1}{1-2\varepsilon_1}$, we have
$
\lim\limits_{k\to\infty}\frac{2 c(2+\eta)(k+1)^{1+\eta}}{\Big(1+4h k^{2+\eta}-h\Delta\Big)^{\frac{1}{2}+\varepsilon_1}}
=\frac{c(2+\eta)}{4^{\varepsilon_1}h^{\frac{1}{2}+\varepsilon_1}},
$
which further gives}
{\setlength\abovedisplayskip{3pt}
\setlength\belowdisplayskip{3pt}
\bna\label{L4-12}
\lim\limits_{k\to\infty}\int_{\frac{T_{k}}{h}}^{\frac{T_{k+1}}{h}}\vartheta(t)dt
\leq\frac{c(2+\eta)}{4^{\varepsilon_1}h^{\frac{1}{2}+\varepsilon_1}}.
\ena
On the other hand, from the monotonicity of $\vartheta(t)$, we have}
{\setlength\abovedisplayskip{3pt}
\setlength\belowdisplayskip{3pt}
\begin{align*}
\int_{\frac{T_{k}}{h}}^{\frac{T_{k+1}}{h}}\vartheta(t)dt
\geq&~
\frac{c}{2\Big(1+\lfloor\frac{T_{k+1}}{h\Delta}\rfloor h\Delta \Big)^{\frac{1}{2}+\varepsilon_1}}\left(4(k+1)^{2+\eta}-4k^{2+\eta}\right)\cr
\geq&~\frac{c}{2\Big(1+4(k+1)^{2+\eta}
h\Big)^{\frac{1}{2}+\varepsilon_1}}\left(4(k+1)^{2+\eta}-4 k^{2+\eta}\right)\cr
\geq&~\frac{2 c(2+\eta)k^{1+\eta}}{\Big(1+4h(k+1)^{2+\eta}\Big)^{\frac{1}{2}+\varepsilon_1}},
\end{align*}
which further leads to
$
\lim\limits_{k\to\infty}\int_{\frac{T_{k}}{h}}^{\frac{T_{k+1}}{h}}\vartheta(t)dt
\geq\frac{c(2+\eta)}{4^{\varepsilon_1}h^{\frac{1}{2}+\varepsilon_1}}.
$
Combining (\ref{L4-12}) and the above inequality yields (iv).}
\vskip 0.2cm

(v) By    $\int_{\frac{T_{k}}{h}}^{\frac{T_{k+1}}{h}}\exp\Big(\int_{\frac{T_{k}}{h}-4}^{s+1}\vartheta(\tau)d\tau\Big)\iota(sh\Delta)ds\leq \exp\Big(\int_{\frac{T_{k}}{h}-4}^{\frac{T_{k+1}}{h}+1}\vartheta(\tau)d\tau\Big)
\int_{\frac{T_{k}}{h}}^{\frac{T_{k+1}}{h}} \iota(sh\Delta)ds$ and (iv),  we know that  $\lim\limits_{k\to\infty}\exp\bigg(\int_{\frac{T_{k}}{h}-4}^{\frac{T_{k+1}}{h}+1}\vartheta(\tau)d\tau\bigg)$ is finite.
Then by the definitions of $\iota(t)$ and $d(t)$, we get
$
\iota(t)=\frac{h\Delta d^2}{(1+k\Delta)^{1+2\varepsilon_2}}, 
~t\in[k\Delta,(k+1)\Delta).
$
Then we have
{\setlength\abovedisplayskip{3pt}
\setlength\belowdisplayskip{3pt}
\begin{align*}
\int_{\frac{T_{k}}{h}}^{\frac{T_{k+1}}{h}}\iota(sh\Delta)ds
\leq&~
\frac{h\Delta d^2}{\Big(1+T_{k}\Delta\Big)^{1+2\varepsilon_2}}\left(\frac{T_{k+1}}{h}-\frac{T_{k}}{h}\right)\cr
\leq&~\frac{h\Delta d^2}{\Big(1+4hk^{2+\eta}\Delta\Big)^{1+2\varepsilon_2}}\left(4(k+1)^{2+\eta}-4k^{2+\eta}\right)\cr
\leq&~\frac{4h\Delta d^2(2+\eta)(k+1)^{1+\eta}}{\Big(1+4hk^{2+\eta}\Delta\Big)^{1+2\varepsilon_2}}.
\end{align*}
From $(2+\eta)(1+2\varepsilon_2)>1+\eta$, we obtain
$
\lim\limits_{k\to\infty}4h\Delta d^2(2+\eta)(k+1)^{1+\eta}\Big(1+4hk^{2+\eta}\Delta\Big)^{-1-2\varepsilon_2}=0,
$
which further gives $\lim\limits_{k\to\infty}\int_{\frac{T_{k}}{h}}^{\frac{T_{k+1}}{h}}\iota(sh\Delta)ds=0.$
Therefore, we have (v). } 
\qed
\vskip 0.2cm

\noindent {\bf Proof of Corollary \ref{coro2}: }
As for any symmetric matrices $A$ and $B$, $\lambda_{\min}(A+B)\geq \lambda_{\min}(A)+\lambda_{\min}(B)$, we have
{\setlength\abovedisplayskip{5pt}
\setlength\belowdisplayskip{5pt}
\begin{align}
&\mu_2\left(-\int_{mh\Delta}^{(m+1)h\Delta}\mathbb{E}[\alpha(s)r(s)+\beta(s)(\mathcal{L}_{\mathcal{G}}\otimes I_n)|\mathcal{F}(mh\Delta)]ds
\right)\notag\\
=&   -\frac{1}{2}\lambda_{\min}\Bigg(\int_{mh\Delta}^{(m+1)h\Delta}\mathbb{E}\big[2\alpha(s)
 r(s)
+\beta(s)
\left(\mathcal{L}_{\mathcal{G}}+\mathcal{L}_{\mathcal{G}}^{\top}\right)\otimes I_n\big|\mathcal{F}(mh\Delta)\big]ds\Bigg)\notag\\
\leq & -\lambda_{\min}\left(\int_{mh\Delta}^{(m+1)h\Delta}\mathbb{E}\left[\alpha(s)
 r(s)\big|\mathcal{F}(mh\Delta)\right]ds\right)\notag\\
&-\frac{1}{2}\lambda_{\min}
\left(\int_{mh\Delta}^{
(m+1)h\Delta}\mathbb{E}\left[\beta(s)
\left(\mathcal{L}_{\mathcal{G}}+\mathcal{L}_{\mathcal{G}}^{\top}\right)\otimes I_n\big|\mathcal{F}(mh\Delta)\right]ds\right).\label{MARKOV0}
\end{align}
For the second term on the r.h.s. of the above inequality, by Condition \ref{condition3} (iii) and (iv), we have}
{\setlength\abovedisplayskip{5pt}
\setlength\belowdisplayskip{5pt}
\begin{align}
&\frac{1}{2}\lambda_{\min}
\left(\int_{mh\Delta}^{
(m+1)h\Delta}\mathbb{E}\left[\beta(s)
\left(\mathcal{L}_{\mathcal{G}}+\mathcal{L}_{\mathcal{G}}^{\top}\right)\otimes I_n\big|\mathcal{F}(mh\Delta)\right]ds\right)
\geq  -\alpha_{2}\frac{h}{2}\Delta \beta(mh\Delta).\label{MARKOV01}
\end{align}
As for any symmetric matrices $A$ and $B$, $\lambda_{\min}(A+B)\geq \lambda_{\min}(A)+\lambda_{\min}(B)$, we have}
{\setlength\abovedisplayskip{5pt}
\setlength\belowdisplayskip{5pt}
\begin{align}
&~\lambda_{\min}\left(\sum_{k=mh}^{(m+1)h-1}\int_{k\Delta}^{(k+1)\Delta}\mathbb{E}
\left[r(s)\big|\mathcal{F}(mh\Delta)\right]ds\right)\notag\\
\geq &~ \lambda_{\min}\left(\sum_{k=mh}^{(m+1)h-1}\int_{k\Delta}^{(k+1)\Delta}
\mathbb{E}\left[r(s)-r(k\Delta)\big|\mathcal{F}(mh\Delta)
\right]ds\right)\notag\\
&~+ \Delta\lambda_{\min}\left(\sum_{k=mh}^{(m+1)h-1}\mathbb{E}\left[
r(k\Delta)\big|\mathcal{F}(mh\Delta)\right]\right).\label{DEPARTMINEIGENVA}
\end{align}
By Definition \ref{exponentialergodic} and the strong 1-exponential ergodicity of $\{r(t),\ t \geq 0\}$ in Condition \ref{condition3} (ii), there exists $R_{1}>0$ and $\delta_{1}>0$, such that  $\sum_{r_{l}\in E}|P\{r(k\Delta)=r_{l}| r(0)=r_{i}\}-\pi_{l}| \leq R_{1} e^{-\delta_{1} k\Delta},\ \forall \ r_{i} \in E$.
By $r_l\leq \alpha_1 I_{Nn}, \ \forall \ r_l\in E$ in Condition \ref{condition3} (ii), we know that $\sup_{r_{l}\in E}\left\|r_{l}\right\|_{2}<\infty$. Then, for the second term on the  r.h.s. of (\ref{DEPARTMINEIGENVA}), similar to the proof of Theorem 2 in \cite{li2018distributed}, we know that there exists $h_{0}>0$, such that if $h \geq h_{0}$, then we have}
{\setlength\abovedisplayskip{5pt}
\setlength\belowdisplayskip{5pt}
\begin{align}
&~ \inf_{m\geq 0}\Delta\lambda_{\min}\left(\sum_{k=mh}^{(m+1)h-1}\mathbb{E}\left[
r(k\Delta)\big|\mathcal{F}(mh\Delta)\right]\right)
\geq~\frac{\Delta}{2} h_{0}\lambda_{\min}\left(\sum_{j=0}^{\infty} \pi_{j}r_{j}\right).\label{DEPARTMINEIGENVA2}
\end{align}
By the Markov property, we have}
{\setlength\abovedisplayskip{5pt}
\setlength\belowdisplayskip{5pt}
\begin{align}
 \mathbb{E}\left[r(s)-r(k\Delta)\big|\mathcal{F}(mh\Delta)\right]
=&~\mathbb{E}\left[\mathbb{E}\left[r(s)-r(k\Delta)\big|\mathcal{F}(k\Delta)\right]\big|\mathcal{F}
(mh\Delta)\right]\notag\\
=&~\mathbb{E}\left[\mathbb{E}\left[r(s)-r(k\Delta)\big|r(k\Delta)
\right]\big|\mathcal{F}(mh\Delta)\right], \notag\\
&~~~~~~~~~~~~~~~~~~~~~\
\forall \ s\in [k\Delta, (k+1)\Delta ].\label{DEPARTMINEIGENVA10}
\end{align}
For any $r_{i}\in E$, suppose $r(k\Delta)=r_{i}$ and denote $T_{i,k}^{\Delta}=\inf\{t\geq k\Delta|r(t)\neq r_{i}\},$ if this set is not empty, otherwise,  $T_{i,k}^{\Delta}=+\infty$. By $\mathbf{0}_{Nn\times Nn}\leq r_{l}\leq \alpha_{1}I_{Nn}, \ \forall \ r_{l}\in E$  in Condition \ref{condition3} (ii)  and Proposition 2.8 in \cite{Anderson1991}, we have}
{\setlength\abovedisplayskip{5pt}
\setlength\belowdisplayskip{5pt}
\begin{align}
\mathbb{E}\left[r(s)-r(k\Delta)\big|r(k\Delta)
\right]
=&\sum_{r_{i}\in E}\mathbb{E}\left[r(s)-r(k\Delta)\big|r(k\Delta)=r_{i}
\right]I_{r(k\Delta)=r_{i}}\notag\\
\geq & \sum_{r_{i}\in E} \Big(P\{T_{i,k}^{\Delta} >(k+1)\Delta|r(k\Delta)=r_{i}\}\times\mathbf{0}_{Nn\times Nn}\notag\\
&+ P\{T_{i,k}^{\Delta} \leq(k+1)\Delta|r(k\Delta)=r_{i}\} \times (-\alpha_{1}I_{Nn})\Big)I_{r(k\Delta)=r_{i}}\notag\\
= & \sum_{r_{i}\in E} \left((1-e^{\gamma_{ii}\Delta})\times (-\alpha_{1}I_{Nn})\right)I_{r(k\Delta)=r_{i}}\notag\\
\geq&  -\alpha_{1}(1-e^{\inf_{r_{i}\in E}\gamma_{ii}\Delta})I_{Nn}, 
\ \forall \ s\in [k\Delta, (k+1)\Delta ].\notag
\end{align}
Then, for the first term on the  r.h.s. of (\ref{DEPARTMINEIGENVA}), by (\ref{DEPARTMINEIGENVA10}), we get}
{\setlength\abovedisplayskip{5pt}
\setlength\belowdisplayskip{5pt}
\begin{align*}
&\lambda_{\min}\bigg(\sum\limits_{k=mh}^{(m+1)h-1} \int_{k\Delta}^{(k+1)\Delta}
\mathbb{E}[r(s)-r(k\Delta)|\mathcal{F}(mh\Delta)
]ds\bigg)
\geq  -h \Delta \alpha_{1}(1-e^{\inf_{r_{i}\in E}\gamma_{ii}\Delta}).\notag
\end{align*}
Then, by the above inequality, $\sup_{r_{i}\in E}|\gamma_{ii}|<\infty$  in Condition  \ref{condition3} (ii), (\ref{DEPARTMINEIGENVA}) and (\ref{DEPARTMINEIGENVA2}), we know that there exists $\Delta_{1}>0$, such that if $0\leq \Delta\leq \Delta_{1}$ and $h \geq h_{0}$, then we have}
{\setlength\abovedisplayskip{5pt}
\setlength\belowdisplayskip{5pt}
\begin{align}
&~\lambda_{\min}\left(\sum_{k=mh}^{(m+1)h-1}\int_{k\Delta}^{(k+1)\Delta}\mathbb{E}
\left[r(s)\Big|\mathcal{F}(mh\Delta)\right]ds\right)
\geq~  \frac{\Delta}{4} h_{0}\lambda_{\min}\left(\sum_{j=0}^{\infty} \pi_{j}r_{j}\right)>0.\notag
\end{align}
By the above inequality and Condition \ref{condition3} (ii) and (iv), we have}
{\setlength\abovedisplayskip{5pt}
\setlength\belowdisplayskip{5pt}
\begin{align*}
 &~\lambda_{\min}\left(\int_{mh\Delta}^{(m+1)h\Delta}\mathbb{E}[\alpha(s)
r(s)|\mathcal{F}(mh\Delta)]ds\right)
\geq~   \frac{\Delta}{4} h_{0} \lambda_{\min}\left(\sum_{j=0}^{\infty} \pi_{j}r_{j} \right)\alpha((m+1)h\Delta).
\end{align*}
Denote $c(m)=\frac{\Delta}{4} h_{0}\lambda_{\min}\big(\sum_{j=0}^{\infty} \pi_{j} r_{j}\big)\alpha((m+1)h\Delta)-\frac{h}{2}\alpha_{2}\Delta \beta(mh\Delta)$.
Then, by (\ref{MARKOV0}), (\ref{MARKOV01}) and the above inequality, we have
$
\mu_2\Big(-\int_{mh\Delta}^{(m+1)h\Delta} \mathbb{E}[\alpha(s)r(s)+\beta(s)(\mathcal{L}_{\mathcal{G}}\otimes I_n)|\mathcal{F}
(mh\Delta)]ds\Big)
\leq -c(m), \notag
$
where $ 0\leq \Delta\leq \Delta_{1}, \ h \geq h_{0}$, and by Condition \ref{condition3} (iv), we have
$\liminf_{m\to\infty}c(m)(1+mh\Delta)^{\frac{1}{2}+\varepsilon_1}>0$.
By Condition \ref{condition3} (ii), we have
$\|\mathcal H(t)\|_2=\sqrt{\lambda_{\max}(\mathcal H^{\top}(t)\mathcal H(t))}\leq\sqrt{\alpha_1}$.}
Then the conditions in Theorem \ref{asycontheorem} are satisfied, so the algorithm converges in mean square.
\qed



\end{appendices}

\end{CJK}
\end{document}